\documentclass[journal,twocolumn]{IEEEtran}
\usepackage{xspace,amsmath,amssymb,amsthm,amsfonts,epsfig,syntonly,dsfont}
\usepackage{cite,bm,color,url,textcomp,array}
\usepackage{algorithmicx,algorithm,algpseudocode}
\usepackage{epstopdf,mathrsfs}
\usepackage{empheq,float}
\usepackage{graphicx,subfigure,balance}
\newtheorem{assumption}{Assumption}
\newtheorem{proposition}{Proposition}

\newtheorem{lemma}{Lemma}
\newtheorem{corollary}{Corollary}
\newtheorem{theorem}{Theorem}

\theoremstyle{definition}\newtheorem{remark}{Remark}

\long\def\symbolfootnote[#1]#2{\begingroup
\def\thefootnote{\fnsymbol{footnote}}
\footnote[#1]{#2}\endgroup}

\psfull

\allowdisplaybreaks[2]

\begin{document}

\title{An Online Convex Optimization Approach to Dynamic Network Resource Allocation}
\author{
Tianyi Chen,~\IEEEmembership{Student Member, IEEE}, Qing Ling,~\IEEEmembership{Senior Member, IEEE},\\ and Georgios B. Giannakis,~\IEEEmembership{Fellow, IEEE}
\thanks {Work in this paper was supported by NSF 1509040, 1508993, 1509005, NSF China 61573331, NSF Anhui 1608085QF130, and CAS-XDA06011203.}
\thanks{T. Chen and G. B. Giannakis are with the Department of Electrical and Computer Engineering and the Digital Technology Center, University of Minnesota, Minneapolis, MN 55455 USA. Emails: \{chen3827, georgios\}@umn.edu

Q. Ling is with the Department of Automation, University of Science and
Technology of China, Hefei, Anhui 230026, China. Email: qingling@mail.ustc.edu.cn
}
}

\markboth{}{}
\maketitle

\begin{abstract}

Existing approaches to online convex optimization (OCO) make sequential one-slot-ahead decisions, which lead to (possibly adversarial) losses that drive subsequent decision iterates.
Their performance is evaluated by the
so-called \textit{regret} that measures the difference of losses
between the online solution and the \textit{best yet fixed} overall
solution in \textit{hindsight}. The present paper deals with
online convex optimization involving adversarial loss
functions and adversarial constraints, where the constraints are
revealed after making decisions, and can be tolerable to instantaneous
violations but must be satisfied in the long term. Performance
of an online algorithm in this setting is assessed by: i) the
difference of its losses relative to the \textit{best dynamic}
solution with one-slot-ahead information of the loss function and
the constraint (that is here termed \textit{dynamic regret}); and, ii) the
accumulated amount of constraint violations (that is here termed
\textit{dynamic fit}). In this context, a modified
online saddle-point (MOSP) scheme is developed, and proved to
simultaneously yield sub-linear dynamic regret and fit, provided
that the accumulated variations of per-slot minimizers and
constraints are sub-linearly growing with time.
MOSP is also applied to the dynamic network resource allocation task, and it is compared
with the well-known stochastic dual gradient method. Under
various scenarios, numerical experiments demonstrate the
performance gain of MOSP relative to the state-of-the-art.

\end{abstract}


\section{Introduction}

Online convex optimization (OCO) is an emerging methodology for
sequential inference with well documented merits especially when the sequence of convex costs varies in an unknown and possibly adversarial manner
\cite{zinkevich2003,hazan2007,shalev2011}.
Starting from the seminal papers \cite{zinkevich2003} and \cite{hazan2007},
most of the early works evaluate OCO algorithms with a
\textit{static regret}, which measures the difference of costs (a.k.a. losses)
between the online solution and the overall best static solution in hindsight.
If an algorithm incurs static regret that increases
sub-linearly with time, then its performance loss averaged over
an infinite time horizon goes to zero; see also
\cite{shalev2011,hazan2016}, and references therein.

However, static regret is not a comprehensive
performance metric \cite{besbes2015}. Take online parameter
estimation as an example. When the true parameter varies over
time, a static benchmark (time-invariant estimator) itself often performs poorly
so that achieving sub-linear static regret is no longer
attractive. Recent works
\cite{besbes2015,hall2015,jadbabaie2015,mokhtari2016b} extend the
analysis of static regret to that of \textit{dynamic regret},
where the performance of an OCO algorithm is benchmarked by the
best dynamic solution with a-priori information on the one-slot-ahead cost function.
Sub-linear dynamic regret is proved
to be possible, if the dynamic environment changes slow enough for the accumulated variation of either costs or
per-slot minimizers to be sub-linearly increasing with respect to the
time horizon.
When the per-slot costs depend on previous decisions, the so-termed competitive
difference can be employed as an alternative of the static regret \cite{chenadam2016,andrew2013}.

The aforementioned works \cite{besbes2015,chenadam2016,andrew2013,hall2015,jadbabaie2015,mokhtari2016b} deal with dynamic costs focusing on problems with time-invariant
constraints that must be strictly satisfied, but do not allow for instantaneous violations of the constraints.
The \textit{long-term} effect of such instantaneous violations was studied in \cite{mahdavi2012}, where an online algorithm with sub-linear static regret and
sub-linear accumulated constraint violation was also developed.
The regret bounds in
\cite{mahdavi2012} have been improved in
\cite{yu2016} and \cite{paternain2016}.
Decentralized optimization with consensus constraints,
as a special case of having long-term but time-invariant
constraints, has been studied in
\cite{koppel2015,wang2012,shahrampour2016}. Nevertheless, \cite{mahdavi2012,paternain2016,koppel2015,yu2016,wang2012,shahrampour2016} do not deal with OCO under time-varying adversarial constraints.

\begin{table*}\addtolength{\tabcolsep}{0pt}
\centering \caption{A summary of related works on discrete time
OCO}\label{tab:comp} \vspace{-0.3cm}
    \begin{tabular}{ c *{3}{|c}}
    \hline
~~~~~~Reference~~~~~~   & ~~~~~~Type of benchmark~~~~~~ & Long-term constraint & Adversarial constraint   \\ \hline
\cite{zinkevich2003}         & Static and dynamic     & No    & No                                 \\
\cite{hazan2007,shalev2011,hazan2016}    & Static       & No    & No                             \\
\cite{besbes2015,mokhtari2016b,zhang2016oco}    & Dynamic     & No    & No                           \\
\cite{hall2015,jadbabaie2015}    & Dynamic    & No    & No                               \\
\cite{chenadam2016,andrew2013} & Dynamic & No & No \\
\cite{mahdavi2012,yu2016,koppel2015,wang2012}    & Static      & Yes    & No                            \\
 \cite{shahrampour2016} & Dynamic     & Yes    & No                        \\
 This work & Dynamic      & Yes    & Yes                             \\
 \hline
    \end{tabular} \vspace{-0.3cm}
\end{table*}

In this context, the present paper considers OCO with time-varying
constraints that must be satisfied in the long term. Under this
setting, the learner first takes an action without knowing
a-priori either the adversarial cost or the time-varying
constraint, which are revealed by the nature subsequently. Its
performance is evaluated by: i) \textit{dynamic regret} that is
the optimality loss relative to a sequence of instantaneous
minimizers with known costs and constraints; and, ii)
\textit{dynamic fit} that accumulates constraint violations
incurred by the online learner due to the lack of knowledge about
future constraints. We compare the OCO setting here with
those of existing works in Table \ref{tab:comp}.

We further introduce a modified online saddle-point (MOSP) method
in this novel OCO framework, where the learner deals with
time-varying costs as well as time-varying but long-term
constraints. We analytically establish that MOSP simultaneously achieves
sub-linear dynamic regret and fit, provided that the
accumulated variations of both minimizers and constraints grow
sub-linearly. This result provides valuable insights for OCO with
long-term constraints: \textit{When the dynamic environment
comprising both costs and constraints does not change on
average, the online decisions provided by MOSP are as good as the
best dynamic solution over a long time horizon.}

To demonstrate the impact of these results, we further apply the proposed MOSP approach to a dynamic network resource allocation task,
where online management of resources is sought without knowing future network states.
Existing algorithms include first- and second-order methods in the
dual domain
\cite{low1999,xiao2004,ghadimi2013,beck2014,wei2013,chen2016jsac}, which are tailored for time-invariant deterministic formulations. To
capture the temporal variations of network resources, stochastic
formulation of network resource allocation has been extensively
pursued since the seminal work of \cite{tassiulas1992}; see also the
celebrated stochastic dual gradient method in
\cite{neely2010,marques12}. These stochastic
approximation-based approaches assume that the
time-varying costs are i.i.d. or generally samples from a stationary ergodic stochastic process
\cite{robbins1951,nemirovski2009,duchi2012}. However, performance
of most stochastic schemes is established in an asymptotic sense, considering the ensemble of per slot averages or infinite samples across time.
Clearly, stationarity
may not hold in practice, especially when the
stochastic process involves human participation.
Inheriting merits of the OCO framework, the proposed MOSP
approach operates in a fully online mode without requiring non-causal information, and further admits finite-sample performance analysis under a sequence of non-stochastic, or even
adversarial costs and constraints.

Relative to existing works, the main contributions of the
present paper are summarized as follows.
\begin{enumerate}
\item [c1)] We generalize the standard OCO framework with only
adversarial costs in
\cite{zinkevich2003,hazan2007,shalev2011,hazan2016} to account for
both adversarial costs and constraints. Different from the regret
analysis in \cite{mahdavi2012,koppel2015,paternain2016,yu2016,wang2012},
performance here is established relative to the best dynamic
benchmark, via metrics that we term dynamic regret and fit.
\item [c2)] We develop an MOSP algorithm to tackle this novel OCO
problem, and analytically establish that MOSP yields
simultaneously sub-linear dynamic regret and fit, provided that
the accumulated variations of per-slot minimizers and constraints
are sub-linearly growing with time.
 \item [c3)]
Our novel OCO approach is tailored for dynamic resource
allocation tasks, where MOSP is compared with the popular
stochastic dual gradient approach.
Relative to the latter, MOSP remains operational in a broader practical setting without probabilistic
assumptions. Numerical tests demonstrate the gain of
MOSP over state-of-the-art alternatives.
\end{enumerate}

\emph{Outline}. The rest of the paper is organized as follows. The
OCO problem with long-term constraints is formulated, and the
relevant performance metrics are introduced in Section II. The MOSP algorithm and its performance analysis are
presented in Section III.
Application of the novel OCO framework and the MOSP algorithm in
network resource allocation, as well as corresponding numerical
tests, are provided in Section IV. Section V concludes the paper.

\emph{Notation}. $\mathbb{E}$ denotes expectation,
$\mathbb{P}$ stands for probability, $(\cdot)^{\top}$ stands for
vector and matrix transposition, and $\|\mathbf{x}\|$ denotes the $\ell_2$-norm of a vector $\mathbf{x}$. Inequalities for vectors,
e.g., $\mathbf{x} > \mathbf{0}$, are defined entry-wise. The
positive projection operator is defined as $[\mathbf{a}]^+:=\max\{\mathbf{a},\mathbf{0}\}$,
also entry-wise.

\section{OCO with Long-term Time-varying Constraints}\label{sec.LTOCO}

In this section, we introduce the generic OCO formulation with
long-term time-varying constraints, along with pertinent metrics to
evaluate an OCO algorithm.

\subsection{Problem formulation}

We begin with the classical OCO setting, where constraints are
time-invariant and must be strictly satisfied. OCO can be viewed as a repeated game between a learner and nature \cite{shalev2011,zinkevich2003,hazan2007}. Consider that
time is discrete and indexed by $t$. Per slot $t$, a learner
selects an action $\mathbf{x}_t$ from a convex set
${\cal X}\subseteq\mathbb{R}^I$, and subsequently nature chooses a
(possibly adversarial) loss function $f_t(\,\cdot\,):
\mathbb{R}^I\rightarrow \mathbb{R}$ through which the learner incurs a
loss $f_t(\mathbf{x}_t)$. The convex set ${\cal X}$
is a-priori known and fixed over the entire time horizon. Although
this standard OCO setting is appealing to various applications
such as online regression and classification
\cite{zinkevich2003,hazan2007,shalev2011}, it does not account for potential variations of (possibly unknown) constraints, and does
not deal with constraints that can possibly be satisfied in the long
term rather than a slot-by-slot basis.
Online optimization with
time-varying and long-term constraints is well motivated for
applications such as navigation, tracking, and dynamic resource allocation
\cite{neely2010,paternain2016,marques12,chen2017tpds}. Taking resource
allocation as an example, time-varying long-term constraints are
usually imposed to tolerate instantaneous violations when
available resources cannot satisfy user requests, and hence allow
flexible adaptation of online decisions to temporal variations of
resource availability.

To broaden the applicability of OCO to these scenarios, we
consider that per slot $t$, a learner selects an action
$\mathbf{x}_t$ from a known and fixed convex set ${\cal
X}\subseteq\mathbb{R}^I$, and then nature reveals not only a
loss function $f_t(\cdot): \mathbb{R}^I\rightarrow \mathbb{R}$ but
also a time-varying (possibly adversarial) penalty function
$\mathbf{g}_t(\cdot): \mathbb{R}^I\rightarrow \mathbb{R}^I$.
This function leads to a time-varying constraint
$\mathbf{g}_t(\mathbf{x})\leq \mathbf{0}$, which is driven by
the unknown dynamics in various applications, e.g., on-demand data
request arrivals in resource allocation. Different from the known
and fixed set ${\cal X}$, the time-varying constraint
$\mathbf{g}_t(\mathbf{x})\leq \mathbf{0}$ can vary arbitrarily or
even adversarially from slot to slot. It is revealed only after the learner makes
her/his decision, and hence it is hard to be satisfied in every
time slot. Therefore, the goal in this context is to find a sequence
of online solutions $\{\mathbf{x}_t\in {\cal X}\}$ that minimize
the aggregate loss, and ensures that the constraints
$\{\mathbf{g}_t(\mathbf{x}_t)\leq \mathbf{0}\}$ are satisfied in the
long-term on average. Specifically, we aim to solve the following
online optimization problem
\begin{subequations}
\label{eq.prob}
\begin{align}
\min_{\{\mathbf{x}_t\in {\cal X},\forall t\}} ~&\sum_{t=1}^T f_t(\mathbf{x}_t)  \label{eq.proba}\\
\text{s. t.}~~&\sum_{t=1}^T \mathbf{g}_t(\mathbf{x}_t) \leq \mathbf{0}\label{eq.probb}
\end{align}
\end{subequations}
where $T$ is the time horizon, $\mathbf{x}_t\in\mathbb{R}^I$ is
the decision variable, $f_t$ is the cost function,
$\mathbf{g}_t:=[g_t^1,\ldots,g_t^I]^{\top}$ denotes the constraint
function with $i$th entry $g_t^i:\mathbb{R}^I\rightarrow
\mathbb{R}$, and ${\cal X}\in \mathbb{R}^I$ is a convex set. The
formulation \eqref{eq.prob} extends the standard OCO framework to
accommodate adversarial time-varying constraints that must be
satisfied in the long term. Complemented by algorithm development
and performance analysis to be carried in the following sections, the main
contribution of the present paper is incorporation of long-term
and time-varying constraints to markedly broaden the scope of OCO.

\subsection{Performance and feasibility metrics}

Regarding performance of online decisions
$\{\mathbf{x}_t\}_{t=1}^T$, static regret is adopted as a metric by standard OCO schemes, under time-invariant and strictly
satisfied constraints. The static regret measures the difference
between the online loss of an OCO algorithm and that of the best
fixed solution in hindsight
\cite{zinkevich2003,hazan2007,shalev2011}. Extending the
definition of static regret over $T$ slots to accommodate time-varying constraints, it can be written as (see also
\cite{paternain2016})
\begin{align}\label{eq.sta-reg}
    {\rm Reg}_{T}^{\rm s}:=\sum_{t=1}^T f_t(\mathbf{x}_t)-\sum_{t=1}^Tf_t(\mathbf{x}^*)
\end{align}
where the best static solution $\mathbf{x}^*$ is obtained as
\begin{equation}\label{eq.slot-opt}
    \mathbf{x}^*\in\arg\min_{\mathbf{x}\in {\cal X}} \,\sum_{t=1}^T f_t(\mathbf{x})~~\text{s. t.}~\mathbf{g}_t(\mathbf{x}) \leq \mathbf{0},\;\forall t.
\end{equation}
A desirable OCO algorithm in this case is the one yielding a
sub-linear regret \cite{mahdavi2012,yu2016}, meaning ${\rm
Reg}_{T}^{\rm s}=\mathbf{o}(T)$. Consequently, $\lim_{T\rightarrow
\infty}{{\rm Reg}_{T}^{\rm s}}/{T}=0$ implies that the algorithm
is ``on average'' no-regret, or in other words, not worse asymptotically than the
best fixed solution $\mathbf{x}^*$. Though widely used in various
OCO applications, the aforementioned \textit{static regret} metric relies on a rather coarse benchmark, which may be less useful especially in dynamic settings.
For instance, \cite[Example 2]{besbes2015} shows
that the gap between the best static and the best dynamic
benchmark can be as large as ${\cal O}(T)$. Furthermore, since the
time-varying constraint $\mathbf{g}_t(\mathbf{x}_t)\leq
\mathbf{0}$ is not observed before making a decision $\mathbf{x}_t$,
its feasibility can not be checked instantaneously.

In response to the quest for improved benchmarks in this dynamic setup, two metrics
are considered here: \textit{dynamic regret} and \textit{dynamic
fit}. The notion of dynamic regret (also termed tracking regret or
adaptive regret) has been recently introduced in
\cite{besbes2015,hall2015,jadbabaie2015,mokhtari2016b} to offer a
competitive performance measure of OCO algorithms under
time-invariant constraints. We adopt it in the setting of
\eqref{eq.prob} by incorporating time-varying
constraints
\begin{align}\label{eq.dyn-reg}
    {\rm Reg}^{\rm d}_T:=\sum_{t=1}^T f_t(\mathbf{x}_t)-\sum_{t=1}^T f_t(\mathbf{x}_t^*)
\end{align}
where the benchmark is now formed via a sequence of best dynamic
solutions $\{\mathbf{x}_t^*\}$ for the instantaneous cost
minimization problem subject to the instantaneous constraint,
namely
\begin{equation}\label{eq.realtime-prob}
    \mathbf{x}_t^*\in\arg\min_{\mathbf{x}\in {\cal X}} \; f_t(\mathbf{x})~~~\text{s. t.}~\mathbf{g}_t(\mathbf{x}) \leq \mathbf{0}.
\end{equation}
Clearly, the dynamic regret is always larger
than the static regret in \eqref{eq.sta-reg}, i.e., ${\rm
Reg}^{\rm s}_T \leq {\rm Reg}^{\rm d}_T$, because
$\sum_{t=1}^Tf_t(\mathbf{x}^*)$ is always no smaller than
$\sum_{t=1}^Tf_t(\mathbf{x}^*_t)$ according to the definitions of
$\mathbf{x}^*$ and $\mathbf{x}^*_t$.
Hence, a sub-linear dynamic regret implies a sub-linear static regret, but not vice versa.

To ensure feasibility of online decisions, the notion of
\textit{dynamic fit} is introduced to measure the
accumulated violation of constraints; under time-invariant
long-term constraints \cite{mahdavi2012,koppel2015,yu2016} or
under time-varying constraints \cite {paternain2016}.
It is defined as
\begin{align}\label{eq.dyn-fit}
    {\rm Fit}^{\rm d}_T:=\left\|\left[\sum_{t=1}^T \mathbf{g}_t(\mathbf{x}_t)\right]^+\right\|.
\end{align}
Observe that the dynamic fit is zero if the accumulated
violation $\sum_{t=1}^T \mathbf{g}_t(\mathbf{x}_t)$ is entry-wise
less than zero. However, enforcing
$\sum_{t=1}^T \mathbf{g}_t(\mathbf{x}_t) \leq \mathbf{0}$ is
different from restricting $\mathbf{x}_t$ to meet
$\mathbf{g}_t(\mathbf{x}_t)\leq \mathbf{0}$ in each and every slot.
While the latter
readily implies the former, the long-term (aggregate) constraint allows adaptation of
online decisions to the environment dynamics; as a result, it is
tolerable to have $\mathbf{g}_t(\mathbf{x}_t)\geq \mathbf{0}$ and
$\mathbf{g}_{t+1}(\mathbf{x}_{t+1})\leq \mathbf{0}$.

An ideal algorithm in this broader OCO framework is the
one that achieves both sub-linear dynamic regret and sub-linear
dynamic fit. A sub-linear dynamic regret implies ``no-regret''
relative to the clairvoyant dynamic solution on the long-term
average; i.e., $\lim_{T\rightarrow \infty}{{\rm Reg}_{T}^{\rm
d}}/{T}=0$; and a sub-linear dynamic fit indicates that the online
strategy is also feasible on average; i.e., $\lim_{T\rightarrow
\infty}{{\rm Fit}_{T}^{\rm d}}/{T}=0$. Unfortunately, the sub-linear
dynamic regret is not achievable in general, even under the
special case of \eqref{eq.prob} where the time-varying constraint
is absent \cite{besbes2015}.
For this reason, we aim at designing and
analyzing an online strategy that generates a sequence $\{\mathbf{x}_t\}_{t=1}^T$ ensuring sub-linear dynamic regret and fit, under mild conditions that must be satisfied by the cost and constraint variations.

\section{Modified Online saddle-point (MOSP) Method}\label{sec.OSP}

In this section, a modified online saddle-point method is
developed to solve \eqref{eq.prob}, and its performance and feasibility are analyzed using the dynamic regret and fit metrics.

\subsection{Algorithm development}

Consider now the per-slot problem \eqref{eq.realtime-prob}, which
contains the current objective $f_t(\mathbf{x})$, the current
constraint $\mathbf{g}_t(\mathbf{x})\leq \mathbf{0}$, and a
time-invariant constraint set ${\cal X}$. With $\bm{\lambda}\in
\mathbb{R}^N_+$ denoting the Lagrange multiplier associated with
the time-varying constraint, the online (partial) Lagrangian of
\eqref{eq.realtime-prob} can be expressed as
\begin{align}\label{eq.Lam}
{\cal L}_t(\mathbf{x},\bm{\lambda}):=f_t(\mathbf{x})+\bm{\lambda}^{\top}\mathbf{g}_t(\mathbf{x})
\end{align}
where $\mathbf{x}\in {\cal X}$ remains implicit. For the online
Lagrangian \eqref{eq.Lam}, we introduce a modified online saddle
point (MOSP) approach, which takes a modified descent step in the
primal domain, and a dual ascent step at each time slot $t$.
Specifically, given the previous primal iterate $\mathbf{x}_{t-1}$
and the current dual iterate $\bm{\lambda}_t$ at each slot $t$,
the current decision $\mathbf{x}_t$ is the minimizer of the
following optimization problem
\begin{equation}\label{eq.primal}
        \min_{\mathbf{x}\in {\cal X}}\, \nabla^{\top} f_{t-1}(\mathbf{x}_{t-1})\!\left(\mathbf{x}-\mathbf{x}_{t-1}\right)+\bm{\lambda}_t^{\top}\mathbf{g}_{t-1}(\mathbf{x})+\frac{\|\mathbf{x}-\mathbf{x}_{t-1}\|^2}{2\alpha}
\end{equation}
where $\alpha$ is a positive stepsize, and $\nabla
f_{t-1}(\mathbf{x}_{t-1})$ is the gradient\footnote{One can
replace the gradient by one of the sub-gradients when
$f_t(\mathbf{x})$ is non-differentiable. The performance analysis
still holds true for this case.} of primal objective
$f_{t-1}(\mathbf{x})$ at $\mathbf{x}=\mathbf{x}_{t-1}$. After the
current decision $\mathbf{x}_t$ is made, $f_t(\mathbf{x})$ and $\mathbf{g}_t(\mathbf{x})$ are observed, and the dual update takes the form
\begin{equation}\label{eq.dual}
        \bm{\lambda}_{t+1}=\big[\bm{\lambda}_t+\mu \nabla_{\bm{\lambda}}{\cal L}_t(\mathbf{x}_t,\bm{\lambda}_t)\big]^{+}=\big[\bm{\lambda}_t+\mu \mathbf{g}_t(\mathbf{x}_t)\big]^{+}
\end{equation}
where $\mu$ is also a positive stepsize, and
$\nabla_{\bm{\lambda}}{\cal
L}_t(\mathbf{x}_t,\bm{\lambda}_t)=\mathbf{g}_t(\mathbf{x}_t)$ is
the gradient of online Lagrangian \eqref{eq.Lam} with respect to
(w.r.t.) $\bm{\lambda}$ at $\bm{\lambda}=\bm{\lambda}_t$.

\begin{remark}
The primal gradient step of the classical saddle-point approach in
\cite{koppel2015,mahdavi2012,paternain2016} is tantamount to
minimizing a first-order approximation of ${\cal
L}_{t-1}(\mathbf{x},\bm{\lambda}_t)$ at
$\mathbf{x}=\mathbf{x}_{t-1}$ plus a proximal term
$\|\mathbf{x}-\mathbf{x}_{t-1}\|^2/(2\alpha)$. We call the
primal-dual recursion \eqref{eq.primal} and \eqref{eq.dual} as a
modified online saddle-point approach, since the primal update
\eqref{eq.primal} is not an exact gradient step when the
constraint $\mathbf{g}_t(\mathbf{x})$ is nonlinear w.r.t.
$\mathbf{x}$. However, when $\mathbf{g}_t(\mathbf{x})$ is linear,
\eqref{eq.primal} and \eqref{eq.dual} reduce to the approach in
\cite{koppel2015,mahdavi2012,paternain2016}. Similar to the primal
update of OCO with long-term but time-invariant constraints in
\cite{yu2016}, the minimization in \eqref{eq.primal} penalizes the
exact constraint violation $\mathbf{g}_t(\mathbf{x})$ instead of
its first-order approximation, which improves control
of constraint violations and facilitates performance analysis of MOSP.
\end{remark}

\begin{algorithm}[t]
\caption{Modified online saddle-point (MOSP) method}\label{algo1}
\begin{algorithmic}[1]
\State \textbf{Initialize:} primal iterate $\mathbf{x}_0$, dual iterate $\bm{\lambda}_1$, and proper
stepsizes $\alpha$ and $\mu$.
\For {$t=1,2\dots$}
\State Update primal variable $\mathbf{x}_t$ by solving \eqref{eq.primal}.
\State Observe the current cost $f_t(\mathbf{x})$ and constraint $\mathbf{g}_t(\mathbf{x})$.
\State Update the dual variable $\bm{\lambda}_{t+1}$ via \eqref{eq.dual}.
\EndFor
\end{algorithmic}
\end{algorithm}

\subsection{Performance analysis}\label{sec.regret}
We proceed to show that for MOSP, the dynamic regret in
\eqref{eq.dyn-reg} and the dynamic fit in \eqref{eq.dyn-fit} are
both sub-linearly increasing if the variations of the
per-slot minimizers and the constraints are small enough. Before
formally stating this result, we assume that the following
conditions are satisfied.

\begin{assumption}\label{ass.0}
For every $t$, the cost function $f_t(\mathbf{x})$ and the
time-varying constraint $\mathbf{g}_t(\mathbf{x})$ in
\eqref{eq.prob} are convex.
\end{assumption}

\begin{assumption}\label{ass.1}
For every $t$, $f_t(\mathbf{x})$ has bounded
gradient on ${\cal X}$; i.e., $\|\nabla
f_t(\mathbf{x})\|\leq G,\;\forall \mathbf{x}\in {\cal X}$; and $\mathbf{g}_t(\mathbf{x})$ is
bounded on ${\cal X}$; i.e.,
$\|\mathbf{g}_t(\mathbf{x})\|\leq M,\;\forall \mathbf{x}\in {\cal
X}$.
\end{assumption}

\begin{assumption}\label{ass.2}
The radius of the convex feasible set ${\cal X}$ is bounded; i.e.,
$\|\mathbf{x}-\mathbf{y}\|\leq R,\; \forall
\mathbf{x},\mathbf{y}\in {\cal X}$.
\end{assumption}

\begin{assumption}\label{ass.3}
There exists a constant $\epsilon>0$, and an interior point $\tilde{\mathbf{x}}\in {\cal X}$ such that $\mathbf{g}_t(\tilde{\mathbf{x}}) \leq -\epsilon\mathbf{1},\;\forall t$.
\end{assumption}

Assumption \ref{ass.0} is necessary for regret analysis in the OCO
setting. Assumption \ref{ass.1} bounds primal and dual gradients per slot, which is also typical in OCO
\cite{shalev2011,yu2016,koppel2015,hall2015}.
Assumption \ref{ass.2} restricts the action set to be bounded.
Assumption \ref{ass.3} is Slater's condition, which guarantees
the existence of a bounded Lagrange multiplier
\cite{bertsekas1999}.

Under these assumptions, we are on track to first
provide an upper bound for the dynamic fit.
\begin{theorem}\label{Them1}
Define the maximum variation of consecutive constraints as
\begin{equation}\label{eq.maxgt}
    \bar{V}(\mathbf{g})\!:=\!\max_t V(\mathbf{g}_t),{\rm~with~} V(\mathbf{g}_t)\!:=\!\max_{\mathbf{x}\in {\cal X}}\,\left\|\left[\mathbf{g}_{t+1}(\mathbf{x})\!-\!\mathbf{g}_t(\mathbf{x})\right]^+\right\|
\end{equation}
and assume the slack constant $\epsilon$ in Assumption \ref{ass.3}
to be larger than the maximum variation\footnote{This equivalently
requires $\epsilon:=\min_{i,t}\max_{\mathbf{x}\in{\cal X}}
[-g^i_t(\mathbf{x})]^+ > \max_{\mathbf{x}\in {\cal X}}$
$\big\|\big[\mathbf{g}_{t+1}(\mathbf{x})-\mathbf{g}_t(\mathbf{x})\big]^+\big\|$,
which is valid when the region defined by
$\mathbf{g}_t(\mathbf{x})\leq \mathbf{0}$ is large enough, or, the
trajectory of $\mathbf{g}_t(\mathbf{x})$ is smooth enough across time.};
i.e., $\epsilon> \bar{V}(\mathbf{g})$. Then under Assumptions
\ref{ass.0}-\ref{ass.3} and the dual variable initialization
$\bm{\lambda}_1 = \mathbf{0}$, the dual iterate for the MOSP recursion
\eqref{eq.primal}-\eqref{eq.dual} is
bounded by
\begin{equation}\label{eq.upper-dual}
    \|\bm{\lambda}_t\|\leq \|\bar{\bm{\lambda}}\|:=\mu M+\frac{2 GR+{R^2}/(2\alpha)+(\mu M^2)/{2}}{\epsilon-\bar{V}(\mathbf{g})},\;\forall t
\end{equation}
and the dynamic fit in \eqref{eq.dyn-fit} is upper-bounded by
    \begin{align}\label{Them.fit}
        {\rm Fit}^{\rm d}_T \leq \frac{\|\bm{\lambda}_{T+1}\|}{\mu} \leq \frac{\|\bar{\bm{\lambda}}\|}{\mu} = M +\frac{2 GR/\mu\!+\!{R^2}/(2\alpha\mu)\!+\!M^2/2}{\epsilon-\bar{V}(\mathbf{g})}
    \end{align}
    where $G$, $M$, $R$, and $\epsilon$ are as in Assumptions \ref{ass.1}-\ref{ass.3}.
\end{theorem}
\begin{proof}
See Appendix \ref{app.thm1}. 	
\end{proof}

Theorem \ref{Them1} asserts that under a mild condition on the
time-varying constraints, $\|\bm{\lambda}_t\|$ is uniformly upper-bounded, and more
importantly, its scaled version ${\|\bm{\lambda}_{T+1}\|}/{\mu}$
upper bounds the dynamic fit. Observe that with a
fixed primal stepsize $\alpha$, ${\rm Fit}^{\rm d}_T$
is in the order of ${\cal O}(1/\mu)$, thus a larger dual stepsize
essentially enables a better satisfaction of long-term
constraints. In addition, a smaller $\bar{V}(\mathbf{g})$ leads to a smaller dynamic fit, which also
makes sense intuitively.

In the next theorem, we further bound the
dynamic regret.

\begin{theorem}\label{Them2}
Under Assumptions \ref{ass.0}-\ref{ass.3} and the dual variable
initialization $\bm{\lambda}_1=\mathbf{0}$, the MOSP recursion
\eqref{eq.primal}-\eqref{eq.dual} yields a dynamic regret
\begin{align}\label{Them.dyn-reg}
    {\rm Reg}^{\rm d}_T\!\leq\! \frac{R V(\{\mathbf{x}_t^*\}_{t=1}^T)}{\alpha}&+\!\|\bar{\bm{\lambda}}\|V(\{\mathbf{g}_t\}_{t=1}^T)\!+\!\frac{R^2}{2\alpha}\nonumber\\
    &+\frac{\alpha G^2T}{2}+\!\frac{\mu M^2(T+1)}{2}
\end{align}
where $V(\{\mathbf{x}_t^*\}_{t=1}^T)$ is the accumulated
variation of the per-slot minimizers $\mathbf{x}^*_t$ defined as
\begin{align}
    V(\{\mathbf{x}_t^*\}_{t=1}^T):=\sum_{t=1}^T \|\mathbf{x}^*_t-\mathbf{x}^*_{t-1}\|
\end{align}
and $V(\{\mathbf{g}_t\}_{t=1}^T)$ is the accumulated
variation of consecutive constraints
\begin{align}\label{eq.var-gt}
   \!\!\!V(\{\mathbf{g}_t\}_{t=1}^T)\!:=\!\sum_{t=1}^T\! V(\mathbf{g}_t)\!=\!\sum_{t=1}^T\max_{\mathbf{x}\in {\cal X}}\! \left\|\left[\mathbf{g}_{t+1}(\mathbf{x})\!-\!\mathbf{g}_t(\mathbf{x})\right]^+\right\|\!.\!
\end{align}
\end{theorem}
\begin{proof}
See Appendix \ref{app.thm2}. 	
\end{proof}
Theorem \ref{Them2} asserts that MOSP's dynamic regret is upper-bounded by a
constant depending on the accumulated variations of per-slot
minimizers and time-varying constraints as well as the primal and
dual stepsizes. While the dynamic regret in the current form
\eqref{Them.dyn-reg} is hard to grasp, the next corollary shall
demonstrate that ${\rm Reg}^{\rm d}_T$ can be very small.

Based on Theorems \ref{Them1}-\ref{Them2}, we are ready to
establish that under the mild conditions for the accumulated
variation of constraints and minimizers, the dynamic regret and
fit are sub-linearly increasing with $T$.
\begin{corollary}\label{ref.coro1}
    Under Assumptions \ref{ass.0}-\ref{ass.3} and the dual variable
initialization $\bm{\lambda}_1=\mathbf{0}$, if the primal and dual stepsizes are chosen such that $\alpha=\mu={\cal O}(T^{-\frac{1}{3}})$, then the dynamic fit is upper-bounded by
    \begin{equation}\label{eq.subfit}
      \! {\rm Fit}^{\rm d}_T\!=\!{\cal O}\!\left(\! M \!+\!\frac{2 GRT^{1/3} \!+\!{R^2T^{2/3}}/2\!+\!M^2/2}{\epsilon-\bar{V}(\mathbf{g})}\!\right)\!\!=\!{\cal O}(T^{\frac{2}{3}}).\!\!\!\!
    \end{equation}
    In addition, if the temporal variations of optimal arguments and constraints satisfy $V(\{\mathbf{x}_t^*\}_{t=1}^T)=\mathbf{o}(T^{\frac{2}{3}})$ and $V(\{\mathbf{g}_t\}_{t=1}^T)=\mathbf{o}(T^{\frac{2}{3}})$, then the dynamic regret is sub-linearly increasing, i.e.,
    \begin{equation}
         {\rm Reg}^{\rm d}_T=\mathbf{o}(T).
    \end{equation}
\end{corollary}
\begin{proof}
    Plugging $\alpha=\mu={\cal O}(T^{-\frac{1}{3}})$ into \eqref{Them.fit}, the bound in \eqref{eq.subfit} readily follows. Likewise, we have from \eqref{Them.dyn-reg} that
    \begin{align}\label{Coro.dyn-reg}
    {\rm Reg}^{\rm d}_T={\cal O}\Big( R V(\{\mathbf{x}_t^*\}_{t=1}^T)T^{\frac{1}{3}}&+\!\|\bar{\bm{\lambda}}\|V(\{\mathbf{g}_t\}_{t=1}^T)\!+\!\frac{R^2T^{\frac{1}{3}}}{2}\nonumber\\
    &+\frac{G^2T^{\frac{2}{3}}}{2}+\!\frac{ M^2T^{\frac{2}{3}}}{2}\Big).
\end{align}
Considering the upper bound on the dual iterates in
\eqref{eq.upper-dual}, it follows that
$\|\bar{\bm{\lambda}}\|={\cal O}(\mu+1/\alpha)={\cal O}(T^{\frac{1}{3}})$, which
implies that
\begin{equation}
    {\rm Reg}^{\rm d}_T\!=\!{\cal O}\left(\max\left\{V(\{\mathbf{x}_t^*\}_{t=1}^T)T^{\frac{1}{3}},\,V(\{\mathbf{g}_t\}_{t=1}^T)T^{\frac{1}{3}},\,T^{\frac{2}{3}}\right\}\right).\!\!
\end{equation}
Therefore, we deduce that ${\rm Reg}^{\rm d}_T=\mathbf{o}(T)$, if
$V(\{\mathbf{x}_t^*\}_{t=1}^T)=\mathbf{o}(T^{\frac{2}{3}})$ and
$V(\{\mathbf{g}_t\}_{t=1}^T)=\mathbf{o}(T^{\frac{2}{3}})$.
\end{proof}

Observe that the sub-linear regret and fit in Corollary \ref{ref.coro1} are achieved under a slightly ``strict'' condition that $V(\{\mathbf{x}_t^*\}_{t=1}^T)=\mathbf{o}(T^{\frac{2}{3}})$ and
$V(\{\mathbf{g}_t\}_{t=1}^T)=\mathbf{o}(T^{\frac{2}{3}})$. The next corollary shows that this condition can be further relaxed if a-priori knowledge of the time-varying environment is available. 
\begin{corollary}\label{ref.coro2}
  Consider Assumptions \ref{ass.0}-\ref{ass.3} are satisfied, and the dual variable
 is initialized as $\bm{\lambda}_1=\mathbf{0}$. If there exists a constant $\beta\in[0,1)$ such that the temporal variations satisfy $V(\{\mathbf{x}_t^*\}_{t=1}^T)=\mathbf{o}(T^{\beta})$ and $V(\{\mathbf{g}_t\}_{t=1}^T)=\mathbf{o}(T^{\beta})$, then choosing the primal and dual stepsizes as $\alpha=\mu={\cal O}(T^{\frac{\beta-1}{2}})$ leads to the dynamic fit
    \begin{equation}\label{eq.subfit2}
     {\rm Fit}^{\rm d}_T\!=\!{\cal O}\!\left(T^{\frac{1-\beta}{2}} \!+T^{1-\beta}\right)\!=\!{\cal O}(T^{1-\beta})=\mathbf{o}(T)
    \end{equation}
 and the corresponding dynamic regret
    \begin{equation}
         {\rm Reg}^{\rm d}_T=\!{\cal O}\!\left(T^{\frac{\beta+1}{2}} \!+\!T^{\frac{1-\beta}{2}}\!+\!T^{\frac{\beta+1}{2}} \right)\!\!=\!{\cal O}(T^{\frac{\beta+1}{2}})=\mathbf{o}(T).
    \end{equation}
\end{corollary}

Corollary \ref{ref.coro2} provides valuable insights for choosing optimal stepsizes in non-stationary settings.
Specifically, adjusting stepsizes to match the variability of the dynamic environment is the key to achieving the optimal performance in terms of dynamic regret and fit. Intuitively, when the variation of the environment is fast (a larger $\beta$), slowly decaying stepsizes (thus larger stepsizes) can better track the potential changes.

\begin{remark}
Theorems \ref{Them1} and \ref{Them2} are in
the spirit of the recent work
\cite{besbes2015,mokhtari2016b,zhang2016oco}, where the regret
bounds are established with respect to a dynamic benchmark in
either deterministic or stochastic settings. However, \cite{besbes2015,mokhtari2016b,zhang2016oco} do not
account for long-term and time-varying constraints, while the dynamic regret analysis is generalized here to the
setting with long-term constraints. Interesting though, sub-linear
dynamic regret and fit can be achieved when the dynamic
environment consisting of the per-slot minimizer and the
time-varying constraint \textit{does not vary on average}, that is,
$V(\{\mathbf{x}_t^*\}_{t=1}^T)$ and $V(\{\mathbf{g}_t\}_{t=1}^T)$
are sub-linearly increasing over $T$.
 \end{remark}

\subsection{Beyond dynamic regret}

Although the dynamic benchmark in \eqref{eq.dyn-reg} is more
competitive than the static one in \eqref{eq.sta-reg}, it is worth
noting that the sequence of the per-slot minimizer
$\mathbf{x}_t^*$ in \eqref{eq.realtime-prob} is not the
optimal solution to problem \eqref{eq.prob}. Defining the sequence of
optimal solutions to \eqref{eq.prob} as $\{\mathbf{x}_t^{\rm
off}\}_{t=1}^T$, it is instructive to see that computing each
minimizer $\mathbf{x}_t^*$ in \eqref{eq.realtime-prob} only
requires one-slot-ahead information (namely, $f_t(\mathbf{x})$ and
$g_t(\mathbf{x})$), while computing each $\mathbf{x}_t^{\rm off}$
within $\{\mathbf{x}_t^{\rm off}\}_{t=1}^T$ requires information over the entire time horizon (that is, $\{f_t(\mathbf{x})\}_{t=1}^T$ and $\{g_t(\mathbf{x})\}_{t=1}^T$).
For this reason, we use the subscript ``off'' in
$\{\mathbf{x}_t^{\rm off}\}_{t=1}^T$ to emphasize that this
solution comes from offline computation with information over $T$ slots.
Note that for the cases without long-term constraints \cite{besbes2015,hall2015,jadbabaie2015,mokhtari2016b}, the sequence of offline solutions $\{\mathbf{x}_t^{\rm off}\}_{t=1}^T$ coincides with the sequence of per-slot minimizers $\{\mathbf{x}_t^*\}_{t=1}^T$.


Regarding feasibility, $\{\mathbf{x}_t^{\rm off}\}_{t=1}^T$ exactly satisfies the
long-term constraint \eqref{eq.probb}, while the solution of MOSP
satisfies \eqref{eq.probb} on average under mild
conditions (cf. Corollary 1). For optimality, the cost of
the online decisions $\{\mathbf{x}_t\}_{t=1}^T$ attained by MOSP is
further benchmarked by the offline solutions $\{\mathbf{x}_t^{\rm
off}\}_{t=1}^T$. To this end, define MOSP's \emph{optimality gap} as
\begin{subequations}\label{subeq.opt-gap}
    \begin{equation}\label{eq.opt-gap}
    {\rm OptGap}^{\rm off}_T:=\sum_{t=1}^T f_t(\mathbf{x}_t)-\sum_{t=1}^T f_t(\mathbf{x}_t^{\rm off}).
\end{equation}
Intuitively, if $\{\mathbf{x}_t^{\rm
off}\}_{t=1}^T$ are close to
$\{\mathbf{x}_t^*\}_{t=1}^T$, the dynamic regret ${\rm Reg}^{\rm
d}_T$ is able to provide an accurate performance measure in the
sense of ${\rm OptGap}^{\rm off}_T$. Specifically, one can
decompose the optimality gap as
 \begin{equation}\label{eq.fac-gap}
    {\rm OptGap}^{\rm off}_T\!=\!\underbrace{\sum_{t=1}^T f_t(\mathbf{x}_t)\!-\!\sum_{t=1}^T f_t(\mathbf{x}_t^*)}_{{\cal U}_1}+\underbrace{\sum_{t=1}^T f_t(\mathbf{x}_t^*)\!-\!\sum_{t=1}^T f_t(\mathbf{x}_t^{\rm off})}_{{\cal U}_2}
\end{equation}
\end{subequations}
where ${\cal U}_1$ corresponds to the dynamic regret
${\rm Reg}^{\rm d}_T$ in \eqref{eq.dyn-reg} capturing the regret
relative to the sequence of per-slot minimizers with
one-slot-ahead information, and ${\cal U}_2$ is
the difference between the performance of per-slot minimizers and
the offline optimal solutions. Although the second term appears
difficult to quantify, we will show next that ${\cal U}_2$ is
driven by the accumulated variation of the dual functions
associated with the instantaneous problems
\eqref{eq.realtime-prob}.

To this end, consider the dual function of the instantaneous primal
problem \eqref{eq.realtime-prob},
which can be expressed by minimizing the online Lagrangian in
\eqref{eq.Lam} at time $t$, namely \cite{bertsekas1999}
\begin{align}\label{eq.slot-dual}
{\cal D}_t(\bm{\lambda}):=\min_{\mathbf{x}\in{\cal X}}\;{\cal
L}_t(\mathbf{x},\bm{\lambda})=\min_{\mathbf{x}\in{\cal
X}}\;f_t(\mathbf{x})+\bm{\lambda}^{\top}\mathbf{g}_t(\mathbf{x}).
\end{align}
Likewise, the dual function of \eqref{eq.prob} over the entire horizon is
\begin{align}\label{eq.all-dual}
{\cal D}(\bm{\lambda})&:=\min_{\{\mathbf{x}_t\in{\cal X},\forall t\}}\;\sum_{t=1}^T{\cal L}_t(\mathbf{x}_t,\bm{\lambda})\nonumber\\
&\stackrel{(a)}{=}\sum_{t=1}^T\min_{\mathbf{x}_t\in{\cal X}}{\cal L}_t(\mathbf{x}_t,\bm{\lambda})\stackrel{(b)}{=}\sum_{t=1}^T{\cal D}_t(\bm{\lambda})
\end{align}
where equality (a) holds since the minimization is separable
across the summand at time $t$, and equality (b) is due to the
definition of the per-slot dual function in \eqref{eq.slot-dual}.
As the primal problems \eqref{eq.prob} and
\eqref{eq.realtime-prob} are both convex, Slater's condition
in Assumption \ref{ass.3} implies that strong duality holds.
Accordingly, ${\cal U}_2$ in \eqref{eq.fac-gap}
can be written as
\begin{equation}\label{eq.U2-dual}
    \sum_{t=1}^T f_t(\mathbf{x}_t^*)-\sum_{t=1}^T f_t(\mathbf{x}_t^{\rm off})=\sum_{t=1}^T\;\max_{\bm{\lambda}\geq \mathbf{0}}{\cal D}_t(\bm{\lambda})-\max_{\bm{\lambda}\geq \mathbf{0}}\sum_{t=1}^T{\cal D}_t(\bm{\lambda})
\end{equation}
which is the difference between the dual objective of the static
best solution, i.e.,
$\bm{\lambda}^*\!\in\!\arg\max_{\bm{\lambda}\geq
\mathbf{0}}\!\sum_{t=1}^T\!{\cal D}_t(\bm{\lambda})$, and that of
the per-slot best solution for \eqref{eq.slot-dual}, i.e.,
$\bm{\lambda}_t^*\!\in\!\arg\max_{\bm{\lambda}\geq
\mathbf{0}}\!{\cal D}_t(\bm{\lambda})$. Leveraging this special
property of the dual problem, we next establish that ${\cal U}_2$
can be bounded by the variation of the dual function, thus
providing an estimate of the optimality gap \eqref{eq.opt-gap}.

\begin{proposition}\label{prop1}
    Define the variation of the dual function \eqref{eq.slot-dual} from time $t$ to $t+1$ as
    \begin{equation}\label{eq.var-dual}
            V({\cal D}_t):=\max_{\bm{\lambda}\geq \mathbf{0}} \|{\cal D}_{t+1}(\bm{\lambda})-{\cal D}_t(\bm{\lambda})\|
    \end{equation}
    and the total variation over the time horizon $T$ as $V(\{{\cal D}_t\}_{t=1}^T):=\sum_{t=1}^T V({\cal D}_t)$. Then
the cost difference between the best offline solution and
the best dynamic solution satisfies
    \begin{equation}
    \sum_{t=1}^T f_t(\mathbf{x}_t^*)-\sum_{t=1}^T f_t(\mathbf{x}_t^{\rm off})\leq 2T V(\{{\cal D}_t\}_{t=1}^T)
    \end{equation}
    where $\mathbf{x}_t^*$ is the minimizer of the instantaneous problem \eqref{eq.realtime-prob}, and $\mathbf{x}_t^{\rm off}$ solves \eqref{eq.prob} with all future information available. Combined with \eqref{eq.fac-gap}, it readily follows that
    \begin{equation}\label{eq.prop1}
    {\rm OptGap}^{\rm off}_T\leq {\rm Reg}^{\rm d}_T+2T V(\{{\cal D}_t\}_{t=1}^T)
    \end{equation}
    where ${\rm Reg}^{\rm d}_T$ is defined in \eqref{eq.dyn-reg}, and ${\rm OptGap}^{\rm off}_T$ in \eqref{subeq.opt-gap}.
\end{proposition}

\begin{proof}
Instead of going to the primal domain, we upper bound
${\cal U}_2$ via the dual representation in \eqref{eq.U2-dual}.
Letting $\tilde{t}$ denote any slot in ${\cal T}:=\{1,\ldots,T\}$, we
have
    \begin{align}\label{eq.proof-prop1}
&\sum_{t\in{\cal T}}\;\max_{\bm{\lambda}\geq \mathbf{0}}{\cal D}_t(\bm{\lambda})-\max_{\bm{\lambda}\geq \mathbf{0}}\sum_{t\in{\cal T}}{\cal D}_t(\bm{\lambda})\\
\leq &\sum_{t\in{\cal T}} \left({\cal D}_t(\bm{\lambda}_t^*)-{\cal D}_t(\bm{\lambda}_{\tilde{t}}^*)\right)\leq T \max_{t\in {\cal T}}\left\{{\cal D}_t(\bm{\lambda}_t^*)\!-\!{\cal D}_t(\bm{\lambda}_{\tilde{t}}^*)\right\}.\nonumber
\end{align}
The first inequality comes from the definition
$\bm{\lambda}_t^*\!\in\!\arg\max_{\bm{\lambda}\geq
\mathbf{0}}\!{\cal D}_t(\bm{\lambda})$. Note that if $\max_{t\in
{\cal T}}\{{\cal D}_t(\bm{\lambda}_t^*)\!-\!{\cal
D}_t(\bm{\lambda}_{\tilde{t}}^*)\} \leq 2 V(\{{\cal
D}_t\}_{t=1}^T)$, the proposition readily follows from
\eqref{eq.proof-prop1}. We will prove this
inequality by contradiction. Assume there exists a slot $t_0\in
{\cal T}$ such that ${\cal D}_{t_0}(\bm{\lambda}_{t_0}^*)-{\cal
D}_{t_0}(\bm{\lambda}_{\tilde{t}}^*)>2V(\{{\cal D}_t\}_{t=1}^T)$,
which implies that
    \begin{align}\label{eq.prop1-cond}
         {\cal D}_{\tilde{t}}(\bm{\lambda}_{\tilde{t}}^*)&\stackrel{(a)}{\leq} {\cal D}_{t_0}(\bm{\lambda}_{\tilde{t}}^*) + V(\{{\cal D}_t\}_{t=1}^T)\stackrel{(b)}{<} {\cal D}_{t_0}(\bm{\lambda}_{t_0}^*)- V(\{{\cal D}_t\}_{t=1}^T)\nonumber\\
         &\stackrel{(c)}{\leq} {\cal D}_{\tilde{t}}(\bm{\lambda}_{t_0}^*),\;\forall\; \tilde{t}\in {\cal T}
    \end{align}
    where inequalities (a) and (c) come from the fact that $V(\{{\cal D}_t\}_{t=1}^T)$ is the accumulated variation over $T$ slots, and hence $\max_{t_1,t_2\in{\cal T}}\|{\cal D}_{t_1}(\bm{\lambda})-{\cal D}_{t_2}(\bm{\lambda})\| \leq V(\{{\cal D}_t\}_{t=1}^T)$, while (b) is due to the hypothesis above.
    Note that ${\cal D}_{\tilde{t}}(\bm{\lambda}_{\tilde{t}}^*)<{\cal D}_{\tilde{t}}(\bm{\lambda}_{t_0}^*)$ in \eqref{eq.prop1-cond} contradicts the fact that $\bm{\lambda}_{\tilde{t}}^*$ is the maximizer of ${\cal D}_{\tilde{t}}(\bm{\lambda})$.
Therefore, we have ${\cal D}_t(\bm{\lambda}_{\tilde{t}}^*)-{\cal D}_t(\bm{\lambda}_t^*)\!\leq\! 2 V(\{{\cal D}_t\}_{t=1}^T)$, which completes the proof.
\end{proof}

The following remark provides an approach to improving the bound
in Proposition 1.

\begin{remark}
Although the optimality gap in \eqref{eq.prop1} appears to be at
least linear w.r.t. $T$, one can use the ``restarting'' trick for
dual variables, similar to that for primal variables in the
unconstrained case; see e.g., \cite{besbes2015}. Specifically, if
the total variation $V(\{{\cal D}_t\}_{t=1}^T)$ is known a-priori,
one can divide the entire time horizon ${\cal
T}:=\{1,\ldots,T\}$ into $\lceil T/\Delta_T\rceil$ sub-horizons
(each with $\Delta_T=\mathbf{o}\big(T/V(\{{\cal
D}_t\}_{t=1}^T)\big)$ slots), and restart the dual iterate
$\bm{\lambda}$ at the beginning of each sub-horizon.
By assuming that $V(\{{\cal
D}_t\}_{t=1}^T)$ is sub-linear w.r.t. $T$, one can guarantee that $\Delta_T\geq 1$ always exists.
In this case, the bound in \eqref{eq.prop1} can be improved by
\begin{equation}
{\rm OptGap}^{\rm
off}_T\leq\lceil T/\Delta_T\rceil {\rm Reg}^{\rm
d}_{\Delta_T}+2\Delta_T  V(\{{\cal D}_t\}_{t=1}^T)	
\end{equation}
where the two summands are
sub-linear w.r.t. $T$ provided that ${\rm Reg}^{\rm d}_{\Delta T}$ over each sub-horizon is sub-linear; i.e., ${\rm Reg}^{\rm
d}_{\Delta_T}=\mathbf{o}(\Delta_T)$.
Interested readers are referred
to \cite{besbes2015} for details of this restarting trick, which
are omitted here due to space limitation.
\end{remark}

\section{Application to network resource allocation}

In this section, we solve the network resource allocation problem
within the OCO framework, and present numerical experiments to
demonstrate the merits of our MOSP solver.

\subsection{Online network resource allocation}\label{subsec.ONRL}

Consider the resource allocation problem over a cloud network
\cite{chen2017tpds}, which is represented by a directed
graph ${\cal G}=({\cal I},\,{\cal E})$ with node set ${\cal I}$
and edge set ${\cal E}$, where $|{\cal I}| = I$ and $|{\cal E}| = E$.
Nodes considered here include mapping nodes collected in the set
${\cal J}=\{1,\ldots,J\}$, and data centers collected in the set
${\cal K}=\{1,\ldots,K\}$; i.e., we have ${\cal I}={\cal J}\bigcup
{\cal K}$.

\begin{figure}[t]
\centering
\includegraphics[width=0.5\textwidth]{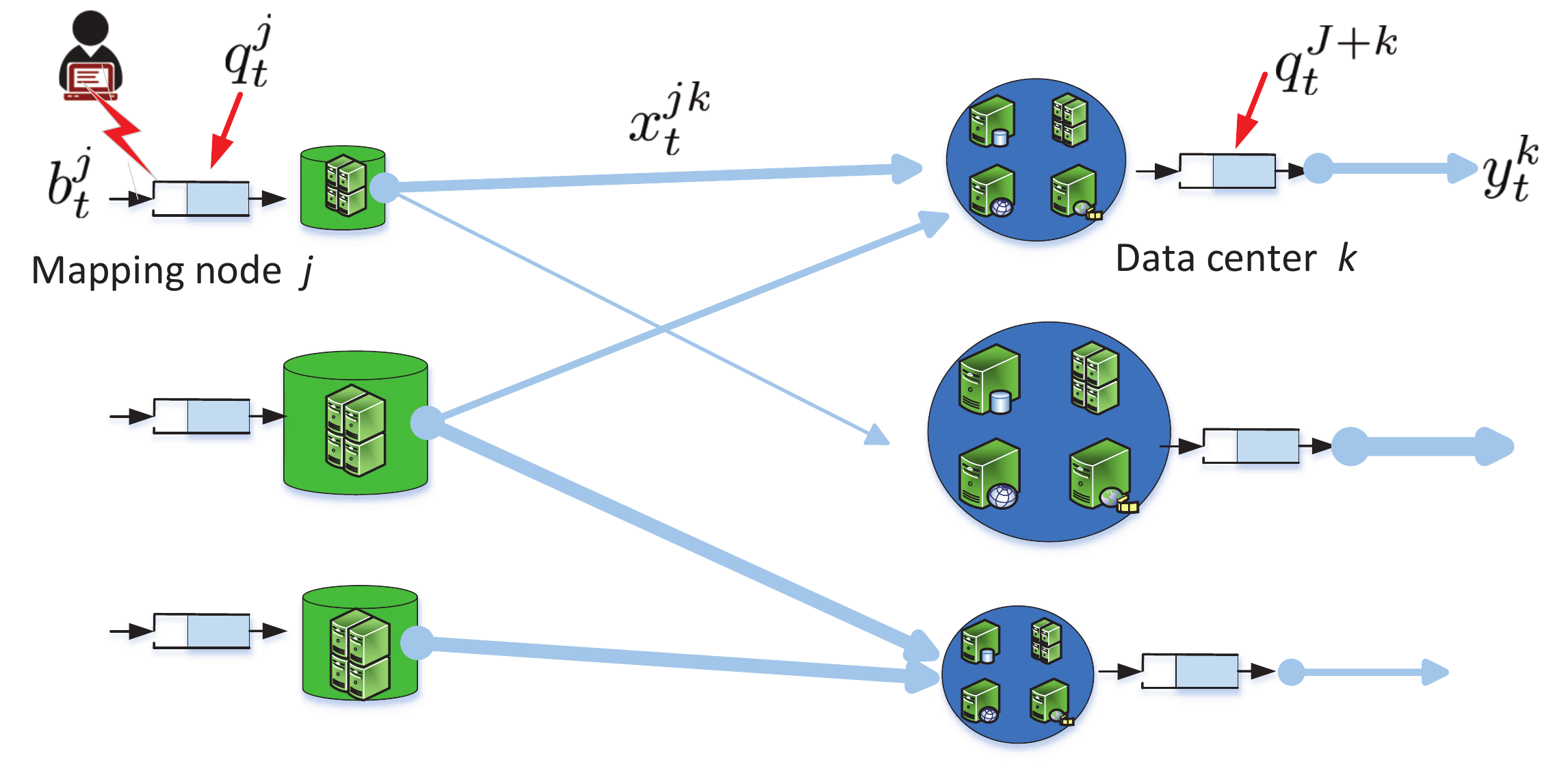}
\vspace{-0.5cm} \caption{A diagram of online network resource
allocation.
Per time $t$, mapping node $j$ has an exogenous workload
$b_t^j$ plus that stored in the queue $q_t^j$, and schedules
workload $x_t^{jk}$ to data center $k$. Data center $k$ serves an amount of workload
$y_t^k$ out of the assigned $\sum_{j=1}^J x_t^{jk}$ as well as that stored in
its queue $q_t^{J+k}$. The thickness of each edge is proportional to its capacity.} \vspace{-0.2cm} \label{fig:system}
\end{figure}

Per time $t$, each mapping node $j$ receives an exogenous
data request $b_t^j$, and forwards the amount $x_t^{jk}$ to each
data center $k$ in accordance with bandwidth availability. Each data
center $k$ schedules workload $y_t^k$ according to its resource
availability. Regarding $y_t^k$ as the weight of a virtual
outgoing edge $(k,*)$ from data center $k$, edge set ${\cal
E}:=\{(j,k),\forall j\in{\cal J},k\in{\cal K}\}\bigcup
\{(k,*),\forall k\in{\cal K}\}$ contains all the links connecting
mapping nodes with data centers, and all the ``virtual'' edges
coming out of the data centers. The $I\times E$ node-incidence
matrix is formed with the $(i,e)$-th entry
   \begin{equation}
    \mathbf{A}_{(i,e)}=
    \left\{
    \begin{array}{rl}
         {1,}~  &\text{if link $e$ enters node $i$}\\
         {-1,}~ &\text{if link $e$ leaves node $i$}\\
         {0,}~ &\text{else.}
    \end{array}
   \right.
\end{equation}
For compactness, collect the data workloads across edges
$e=(i,j)\in {\cal E}$ in a resource allocation vector
$\mathbf{x}_t:=[x_t^{11},\ldots,x_t^{JK},y_t^1,\ldots,y_t^K]^{\top}\in
\mathbb{R}^{E}_+$, and the exogenous load arrival rates of
all nodes in a vector
$\mathbf{b}_t:=[b_t^1,\ldots,b_t^J,0\ldots,0]^{\top}\in
\mathbb{R}_+^{I}$. Then, the aggregate (endogenous plus exogenous)
workloads of all nodes are given by
$\mathbf{A}\mathbf{x}_t+\mathbf{b}_t$. When the $i$-th entry of
$\mathbf{A}\mathbf{x}_t+\mathbf{b}_t$ is positive, there is
service residual at node $i$; otherwise, node $i$ over-serves the
current workload arrival.
Assume that each data center and mapping
node has a local data queue to buffer unserved workloads
\cite{neely2010}.
With $\mathbf{q}_t:=[q_t^1,\ldots,q_t^{J+K}]^{\top}$ collecting the queue lengths at each mapping node and data center, the queue update is
$\mathbf{q}_{t+1}=\left[\mathbf{q}_t+\mathbf{A}\mathbf{x}_t+\mathbf{b}_t\right]^{+}$,
where $[\,\cdot\,]^+$ ensures that the queue length is always
non-negative. The bandwidth limit of link $(j,k)$ is
$\bar{x}^{jk}$, and the resource capability of data center $k$ is
$\bar{y}^k$, which can be compactly expressed by $\mathbf{x} \in
{\cal X}$ with ${\cal X}:=\{\mathbf{0}\leq \mathbf{x}\leq
\bar{\mathbf{x}}\}$ and
$\bar{\mathbf{x}}:=[\bar{x}^{11},\ldots,\bar{x}^{JK},\bar{y}^1,\ldots,\bar{y}^K]^{\top}$.
The overall system diagram is depicted in Fig. \ref{fig:system}.

For each data center, the power cost $f_t^k
(y_t^k):=f^k(y_t^k;\theta_t^k)$ depends on a time-varying
parameter $\theta_t^k$, which captures the energy price
and the renewable generation at data center $k$ during slot
$t$. The bandwidth cost $f_t^{jk}
(x_t^{jk}):=f^{jk}(x_t^{jk};\theta_t^{jk})$ characterizes the
transmission delay and is parameterized by a time-varying scalar
$\theta_t^{jk}$. Scalars $\theta_t^k$ and $\theta_t^{jk}$ can be readily extended to vector
forms. To keep the exposition simple, we use scalars to represent
time-varying factors at nodes and edges.

Per slot $t$, the instantaneous cost
$f_t(\mathbf{x}_t)$ aggregates the costs of power consumed at all
data centers plus the bandwidth costs at all links, namely
\begin{equation}\label{eq.netcost}
f_t(\mathbf{x}_t):=\sum_{k\in{\cal
K}}~\underbrace{~~f_t^k(y_t^k)~~}_{\rm power~cost}~+~\sum_{j\in
{\cal J}}\sum_{k\in{\cal K}}~\underbrace{~~f_t^{jk}
(x_t^{jk})~~}_{\rm bandwidth~cost}
\end{equation}
where the objective can be also written as
$f_t(\mathbf{x}_t):=f(\mathbf{x}_t;\bm{\theta}_t)$ with
$\bm{\theta}_t:=[\theta_t^1,\ldots,\theta_t^K,\theta_t^{11},\ldots,\theta_t^{JK}]^{\top}$
concatenating all time-varying parameters.
Aiming to minimize the accumulated cost while serving all workloads, the optimal workload routing and allocation strategy in
this cloud network is the solution of the following optimization problem
\begin{align}\label{eq.netprob}
\min_{\{\mathbf{x}_t\in {\cal X},\forall t\}} \, \sum_{t=1}^T f_t(\mathbf{x}_t)~~~
\text{s. t.}~~&\mathbf{q}_{t+1}=\left[\mathbf{q}_t+\mathbf{A}\mathbf{x}_t+\mathbf{b}_t\right]^{+},\,\forall t\nonumber\\
&\mathbf{q}_1\geq\mathbf{0},~\mathbf{q}_{T+1}=\mathbf{0}
\end{align}
where $\mathbf{q}_1$ is the given initial queue length, and $\mathbf{q}_{T+1}=\mathbf{0}$ guarantees that all workloads arrived have been served at the end of the scheduling horizon.
Note that \eqref{eq.netprob} is time-coupled, and generally challenging to solve without information of future workload arrivals and time-varying cost functions.
Therefore, we reformulate \eqref{eq.netprob} to fit our OCO formulation \eqref{eq.prob} by
relaxing the queue recursion in \eqref{eq.netprob}, namely
\begin{equation}\label{eq.queue-relax}
    \mathbf{q}_{T+1}\geq \mathbf{q}_T+\mathbf{A}\mathbf{x}_T+\mathbf{b}_T \geq \mathbf{q}_1+\sum_{t=1}^T (\mathbf{A}\mathbf{x}_t+\mathbf{b}_t)
\end{equation}
which readily leads to $\sum_{t=1}^T
(\mathbf{A}\mathbf{x}_t+\mathbf{b}_t)\leq
\mathbf{q}_{T+1}-\mathbf{q}_1\leq \mathbf{0}$, since
$\mathbf{q}_1\geq\mathbf{0}$ and $\mathbf{q}_{T+1}=\mathbf{0}$.
Therefore, instead of solving \eqref{eq.netprob}, we aim to tackle
a relaxed problem that is in the form of OCO with long-term
constraints, given by
\begin{align}\label{eq.apply}
\min_{\{\mathbf{x}_t\in {\cal X},\forall t\}} \, &\sum_{t=1}^T f_t(\mathbf{x}_t)~~~
\text{s. t.}~~\sum_{t=1}^T \left(\mathbf{A}\mathbf{x}_t+\mathbf{b}_t\right) \leq \mathbf{0}
\end{align}
where the workload flow conservation constraint
$\mathbf{A}\mathbf{x}_t+\mathbf{b}_t\leq \mathbf{0}$ must be satisfied in the long term rather than slot-by-slot.
Clearly,
\eqref{eq.apply} is in the form of \eqref{eq.prob}. Therefore,
the MOSP algorithm of Section \ref{sec.OSP} can be
leveraged to solve \eqref{eq.apply} in an \textit{online} fashion, with
provable performance and feasibility guarantees. Specifically, with
$\mathbf{g}_t(\mathbf{x}_t)=\mathbf{A}\mathbf{x}_t+\mathbf{b}_t$,
the primal update \eqref{eq.primal} boils down to a simple
gradient update $\mathbf{x}_t\!=\!{\cal P}_{\cal
X}\left(\mathbf{x}_{t-1}-\alpha \nabla
f_{t-1}(\mathbf{x}_{t-1})-\alpha\mathbf{A}^{\top}\bm{\lambda}_t\right)$,
where ${\cal P}_{\cal X}\left(\cdot\right)$ defines projection
onto the convex set ${\cal X}$. The dual update \eqref{eq.dual} is $\bm{\lambda}_{t+1}\!=\!\big[\bm{\lambda}_t+\mu
(\mathbf{A}\mathbf{x}_t+\mathbf{b}_t)\big]^{+}\!$, which can be nicely regarded as a scaled version of the relaxed queue dynamics in \eqref{eq.netprob}, with $\mathbf{q}_t=\bm{\lambda}_t/\mu$.

In addition to simple closed-form updates, MOSP can also afford a
fully decentralized implementation by exploiting the problem
structure of network resource allocation, where each mapping node
or data center decides the amounts on all its \textit{outgoing
links}, and only exchanges information with its \textit{one-hop
neighbors}.
Per time slot $t$, the primal update at mapping node $j$ includes variables on all its outgoing links, given by
\begin{subequations}\label{eq.dist-net}
\begin{equation}\label{eq.primal-gd1}
   x_t^{jk}\!=\!\left[x_{t-1}^{jk}\!-\!\alpha \nabla f^{jk}_{t-1}(x_{t-1}^{jk})\!-\!\alpha\left(\lambda_t^k\!-\!\lambda_t^j\right)\right]_{0}^{\bar{x}^{jk}}\!\!\!,\;\forall k\in{\cal K}
\end{equation}
and the dual update reduces to
 \begin{equation}\label{eq.dual-gd1}
\lambda_{t+1}^j=\left[\lambda_t^j+\mu \left(b_t^j-\sum_{k\in{\cal
K}}x_t^{jk}\right)\right]^{+}.
\end{equation}
Likewise, for data center $k$, the primal update becomes
\begin{equation}\label{eq.primal-gd2}
    y_t^k=\left[y_{t-1}^k-\alpha \nabla f^{k}_{t-1}(y_{t-1}^k)-\alpha \sum_{j\in{\cal J}}(\lambda_t^k-\lambda_t^j)\right]_0^{\bar{y}^k}
\end{equation}
where $[\;\cdot\;]_0^{\bar{y}^k}\!\!:=\min\{\bar{y}^k,\max\{\cdot\,,0\}\}$, and the dual recursion is
\begin{equation}\label{eq.dual-gd2}
\lambda_{t+1}^k=\left[\lambda_t^k+\mu \left(\sum_{j\in{\cal
J}}x_t^{jk}-y_t^k\right)\right]^{+}.
\end{equation}
\end{subequations}
Distributed MOSP for online network resource allocation is
summarized in Algorithm \ref{algo2}.

\begin{algorithm}[t]
\caption{\!Distributed MOSP for network resource
allocation}\label{algo2}
\begin{algorithmic}[1]
\State \textbf{Initialize:} primal iterate $\mathbf{x}_0$, dual
iterate $\bm{\lambda}_1$, and proper stepsizes $\alpha$ and $\mu$.
\For {$t=1,2\dots$} \State Each mapping node $j$ performs
\eqref{eq.primal-gd1} and each data {\color{white} ccc\,}center $k$ runs \eqref{eq.primal-gd2}.
\State Mapping nodes and data centers observe local costs
{\color{white} ccc\,}and workload arrivals. \State Each mapping node $j$ performs
\eqref{eq.dual-gd1} and each data
{\color{white} ccc\,}center $k$ performs
\eqref{eq.dual-gd2}. \State Mapping nodes (data centers) send
multipliers to all
{\color{white} ccc\,}neighboring data centers (mapping nodes).
\EndFor
\end{algorithmic}
\end{algorithm}

\subsection{Revisiting stochastic dual (sub)gradient}\label{subsec.SDG}

The dynamic network resource allocation problem in Section
\ref{subsec.ONRL} has so far been studied in the stochastic
setting \cite{chen2016,chen2017tpds}. Classical approaches include
Lyapunov optimization \cite{tassiulas1992,neely2010}
and the stochastic dual (sub)gradient method
\cite{marques12}, both of which rely on stochastic
approximation (SA) \cite{robbins1951}. In the context of
stochastic optimization, the time-varying vectors $\{\bm{\xi}_t\}$
with
$\bm{\xi}_t\!:=\![\bm{\theta}_t^{\top},\mathbf{b}_t^{\top}]^{\top}$
appearing in the cost and constraint are assumed to be independent
realizations of a random variable $\bm{\Xi}$.\footnote{Extension is also available when $\{\bm{\xi}_t\}$ constitute a sample
path from an ergodic stochastic process $\{\bm{\Xi}_t\}$, which
converges to a stationary distribution; see e.g.,
\cite{duchi2012,Ale10}.} In an SA-based stochastic optimization
algorithm, per time $t$, a policy first observes a realization
$\bm{\xi}_t$ of the random variable $\bm{\Xi}$, and then
(stochastically) selects an action $\mathbf{x}_t\in {\cal X}$.
However, in contrast to minimizing the \textit{observed cost} in
the OCO setting, the goal of the stochastic resource allocation is
usually to minimize the limiting average of the \textit{expected
cost} subject to the so-termed stability constraint, namely
\begin{subequations}\label{eq.stoc-prob}
\begin{align}
\min_{\{\mathbf{x}_t\in {\cal X},\mathbf{q}_t,\forall t\}} \, &\lim_{T\rightarrow \infty}\frac{1}{T}\sum_{t=1}^T \mathbb{E}[f_t(\mathbf{x}_t)]\label{eq.stoc-proba}\\
\text{s. t.}~~&\mathbf{q}_{t+1}=\left[\mathbf{q}_t+\mathbf{A}\mathbf{x}_t+\mathbf{b}_t\right]^{+}\!,~\forall t\label{eq.stoc-probb}\\
&\lim_{T\rightarrow \infty}\frac{1}{T}\sum_{t=1}^T \mathbb{E}\left[\mathbf{q}_t\right] \leq \mathbf{0}\label{eq.stoc-probc}
\end{align}
\end{subequations}
where he expectation in \eqref{eq.stoc-proba} is taken
over $\bm{\Xi}$ and the randomness of $\mathbf{x}_t$ and $\mathbf{q}_t$ induced by all possible sample paths $\{\bm{\xi}_1,\ldots,\bm{\xi}_t\}$ via \eqref{eq.stoc-probb};
and the stability constraint \eqref{eq.stoc-probc} implies a finite bound on the accumulated constraint violation.
In contrast to the observed costs in \eqref{eq.netprob}, each decision $\mathbf{x}_t$ is evaluated by
all possible realizations in $\bm{\Xi}$ here.
However, as $\mathbf{q}_t$ in \eqref{eq.stoc-probb} couples the optimization
variables over an infinite time horizon, \eqref{eq.stoc-prob} is intractable in general.

Prior works \cite{neely2010,Geor06,marques12,chen2017tpds} have demonstrated
that \eqref{eq.stoc-prob} can
be tackled via a tractable stationary relaxation, given by
\begin{subequations}\label{eq.stoc-relax}
\begin{align}
\min_{\{\mathbf{x}_t\in {\cal X},\forall t\}} \, &\lim_{T\rightarrow \infty}\frac{1}{T}\sum_{t=1}^T \mathbb{E}[f_t(\mathbf{x}_t)]\\
\text{s. t.}~&\lim_{T\rightarrow \infty}\frac{1}{T}\sum_{t=1}^T \mathbb{E}\left[\mathbf{A}\mathbf{x}_t+\mathbf{b}_t\right] \leq \mathbf{0}\label{eq.stoc-relaxb}
\end{align}
\end{subequations}
where the time-coupling constraints \eqref{eq.stoc-probb} and
\eqref{eq.stoc-probc} are relaxed to the limiting average
constraint \eqref{eq.stoc-relaxb}.
Such a relaxation can be verified similar to the queue relaxation in \eqref{eq.queue-relax}; see also \cite{neely2010}.
Note that \eqref{eq.stoc-relax} is still challenging since it involves expectations in both costs and constraints, and the
distribution of $\bm{\Xi}$ is usually unknown. Even if the joint
probability distribution function were available, finding the
expectations would not scale with the dimensionality of $\bm{\Xi}$.
A common remedy is to use the stochastic dual gradient
(SDG) iteration (a.k.a. Lyapunov optimization) \cite{tassiulas1992,neely2010,chen2017tpds}.
Specifically, with $\bm{\lambda}\in \mathbb{R}_+^{I}$ denoting the
multipliers associated with the expectation constraint \eqref{eq.stoc-relaxb}, the SDG method first observes one realization $\bm{\xi}_t$ at each slot $t$, and then performs the dual update as
\begin{align}\label{eq.dual-stocg}
\bm{\lambda}_{t+1} ~=\big[\bm{\lambda}_t+ \mu (\mathbf{A}\mathbf{x}_t+\mathbf{b}_t)\big]^{+},\;\forall t
\end{align}
where $\bm{\lambda}_t$ is the dual iterate at time $t$, $\mathbf{A}\mathbf{x}_t+\mathbf{b}_t$ is the stochastic dual gradient, and $\mu$
is a positive (and typically constant) stepsize.
The actual allocation or the primal variable $\mathbf{x}_t$ appearing in \eqref{eq.dual-stocg} needs be found by solving the following sub-problems, one per slot $t$
\begin{align}\label{eq.SA-sub}
    \mathbf{x}_t\in\arg\min_{\mathbf{x} \in
{\cal X}}f_t(\mathbf{x})+\bm{\lambda}_t^{\top}(\mathbf{A}\mathbf{x}+\mathbf{b}_t).
\end{align}

For the considered network resource allocation problem, SDG in
\eqref{eq.dual-stocg}-\eqref{eq.SA-sub} entails a well-known
cost-delay tradeoff \cite{neely2010}. Specifically, with $f^*$
denoting the optimal objective \eqref{eq.stoc-relax}, SDG can
achieve an ${\cal O}(\mu)$-optimal solution such that
$\lim_{T\rightarrow \infty}({1}/{T}) \sum_{t=1}^{T}
\mathbb{E}\left[f_t\left(\mathbf{x}_t\right)\right] \!\leq\!
f^*\!+{\cal O}(\mu)$, and guarantee queue
lengths\footnote{According to Little's law \cite{little1961}, the time-average delay
is proportional to the time-average queue length given the
arrival rate.} satisfying $\lim_{T\rightarrow \infty}({1}/{T})
\sum_{t=1}^{T} \mathbb{E}\left[\|\mathbf{q}_t\|\right]\!=\!{\cal
O}({1}/{\mu})$. Therefore, reducing the optimality gap ${\cal O}(\mu)$
will essentially increase the average network delay ${\cal O}(1/\mu)$.


\begin{remark}
    The optimality of SDG is established relative to the offline optimal solution of \eqref{eq.stoc-relax}, which can be thought as the time-average \textit{optimality gap} in \eqref{eq.opt-gap} under the OCO setting.
    Interestingly though, the optimality gap under the stochastic setting is equivalent to the (expected) dynamic regret \eqref{eq.dyn-reg}, since their (expected) difference $ V(\{\mathbb{E}[{\cal D}_t]\}_{t=1}^T)$ in \eqref{eq.prop1} reduces to zero.
    To see this, note that $\mathbb{E}[f_t(\mathbf{x})]$ and $\mathbb{E}[\mathbf{A}\mathbf{x}+\mathbf{b}_t]$ are time-invariant, hence the dual problem of each per-slot subproblem in \eqref{eq.stoc-relax} is time-invariant.
    This reduction means that the SDG solver of the dynamic problem in \eqref{eq.stoc-prob} leverages its inherent stationarity (through the stationary dual problem), in contrast to the non-stationary nature of the OCO framework.
\end{remark}


\begin{remark}
Below we highlight several differences of the novel MOSP
in Algorithm \ref{algo2} with the SDG recursion in
\eqref{eq.dual-stocg}-\eqref{eq.SA-sub} for the dynamic network resource
allocation task.

(D1) From an operational perspective, SDG observes the current state
$\bm{\xi}_t$ first, and then performs the resource allocation decision
$\mathbf{x}_t$ accordingly. Therefore, at the beginning of
slot $t$, SDG needs to precisely know the non-causal information $\bm{\xi}_t$. Inheriting the
merits of OCO, on the other hand, MOSP operates in a fully
\textit{predictive} mode, which decides $\mathbf{x}_t$ without
knowing the cost $f_t(\mathbf{x})$ and the constraint
$\mathbf{g}_t(\mathbf{x})$ (or $\bm{\xi}_t$) at time $t$.
This feature of MOSP is of major practical importance when costs and availability of resources are not available at the point of making decisions; e.g., online demand response in smart grids \cite{kim2017} and resource allocation in wireless networking \cite{tadrous2013}.

(D2) From a computational point of view, MOSP reduces to a simple
saddle-point recursion with primal (projected) gradient descent
and dual gradient ascent for the network resource allocation
problem, both of which incur affordable complexity. However, the
primal update of SDG in \eqref{eq.SA-sub} generally requires
solving a convex program per time slot $t$, which leads to much
higher computational complexity in general.

(D3) With regards to the theoretical claims, the time-varying vector
$\bm{\xi}_t$ in SDG typically requires a rather restrictive probabilistic
assumption, to establish SDG optimality in either the
ensemble average \cite{neely2010} or the limiting
ergodic average sense \cite{Ale10}. In contrast, leveraging the
OCO framework, MOSP admits finite-sample performance analysis with non-stochastic observed costs and constraints, which can even be adversarial.
\end{remark}

\subsection{Numerical experiments}

In this section, we provide numerical tests to demonstrate the
merits of the proposed MOSP algorithm in the application of dynamic
network resource allocation. Consider the geographical workload
routing and allocation task in \eqref{eq.apply} with $J=10$
mapping nodes and $K=10$ data centers. The instantaneous
network cost in \eqref{eq.netcost} is
\begin{equation}
f_t(\mathbf{x}_t)\!:=\sum_{k\in{\cal K}}p_t^k (y_t^k)^2+\sum_{j\in {\cal J}}\sum_{k\in{\cal K}}c^{jk}
(x_t^{jk})^2
\end{equation}
where $p_t^k$ is the energy price at data center $k$ at time $t$,
and $c^{jk}$ is the per-unit bandwidth cost for transmitting from
mapping node $j$ to data center $k$. With the bandwidth limit
$\bar{x}^{jk}$ uniformly randomly generated within $[10,100]$, we
set the bandwidth cost of each link $(j,k)$ as
$c^{jk}=40/\bar{x}^{jk},\forall j,k$. The resource capacities
$\{\bar{y}^k,\forall k\}$ at all data centers are uniformly
randomly generated from $[100,200]$. We consider the following two
cases for the time-varying parameters $\{p_t^k,\forall t,k\}$ and
$\{b_t^j,\forall t,j\}$:

  \textbf{Case 1)} Parameters $\{p_t^k,\forall t,k\}$ and $\{b_t^j,\forall t,j\}$ are independently drawn from time-invariant distributions. Specifically, $p_t^k$ is uniformly distributed over $[1,3]$, and the delay-tolerant workload $b_t^j$ arrives at each mapping node $j$ according to a uniform distribution over $[50,150]$.

  \textbf{Case 2)} Parameters $\{p_t^k,\forall t,k\}$ and $\{b_t^j,\forall t,j\}$ are generated according to non-stationary stochastic processes. Specifically, $p_t^k=\sin(\pi t/12)+n_t^k$ with i.i.d. noise $n_t^k$ uniformly distributed over $[1,3]$, while $b_t^j=50\sin(\pi t/12)+v_t^j$ with i.i.d. noise $v_t^j$ uniformly distributed over $[99,101]$.

Finally, with time horizon $T=500$, the stepsize in \eqref{eq.primal-gd1} and \eqref{eq.primal-gd2} is
set to $\alpha=0.05/T^{1/3}$, and for \eqref{eq.dual-gd1} and \eqref{eq.dual-gd2} to $\mu=50/T^{1/3}$.
MOSP is benchmarked by three strategies: SDG in Section
\ref{subsec.SDG}, the sequence of per-slot best minimizers in
\eqref{eq.slot-opt}, and the offline optimal solution that solves
\eqref{eq.prob} at once with all future costs and constraints available.
Note that at the beginning of each slot $t$, the exact prices
$\{p_t^k,\forall k\}$ and demands $\{b_t^j,\forall j\}$ for the
coming slot are generally not available in practice \cite{lmp,kim2017,tadrous2013,huang2016}.
Since the original SDG updates
\eqref{eq.dual-stocg} and \eqref{eq.SA-sub} require non-causal
knowledge of $\{p_t^k,\forall k\}$ and $\{b_t^j,\forall j\}$ to
decide $\mathbf{x}_t$, we modify them for fairness in this online setting by using the prices and demands at slot $t-1$
to obtain $\mathbf{x}_t$. In this case that we we term online dual gradient (ODG), the performance guarantee of SDG may not hold.
Nevertheless, as shown in the next, different constant stepsizes
for ODG's dual update in \eqref{eq.dual-stocg} still lead to quite different performance and feasibility behaviors. For this reason,
ODG is studied under stepsizes $\mu_{\rm ODG}=0.5$ and
$1$.

\begin{figure}[t]
\centering
\vspace{-0.2cm}
\includegraphics[height=0.3\textwidth]{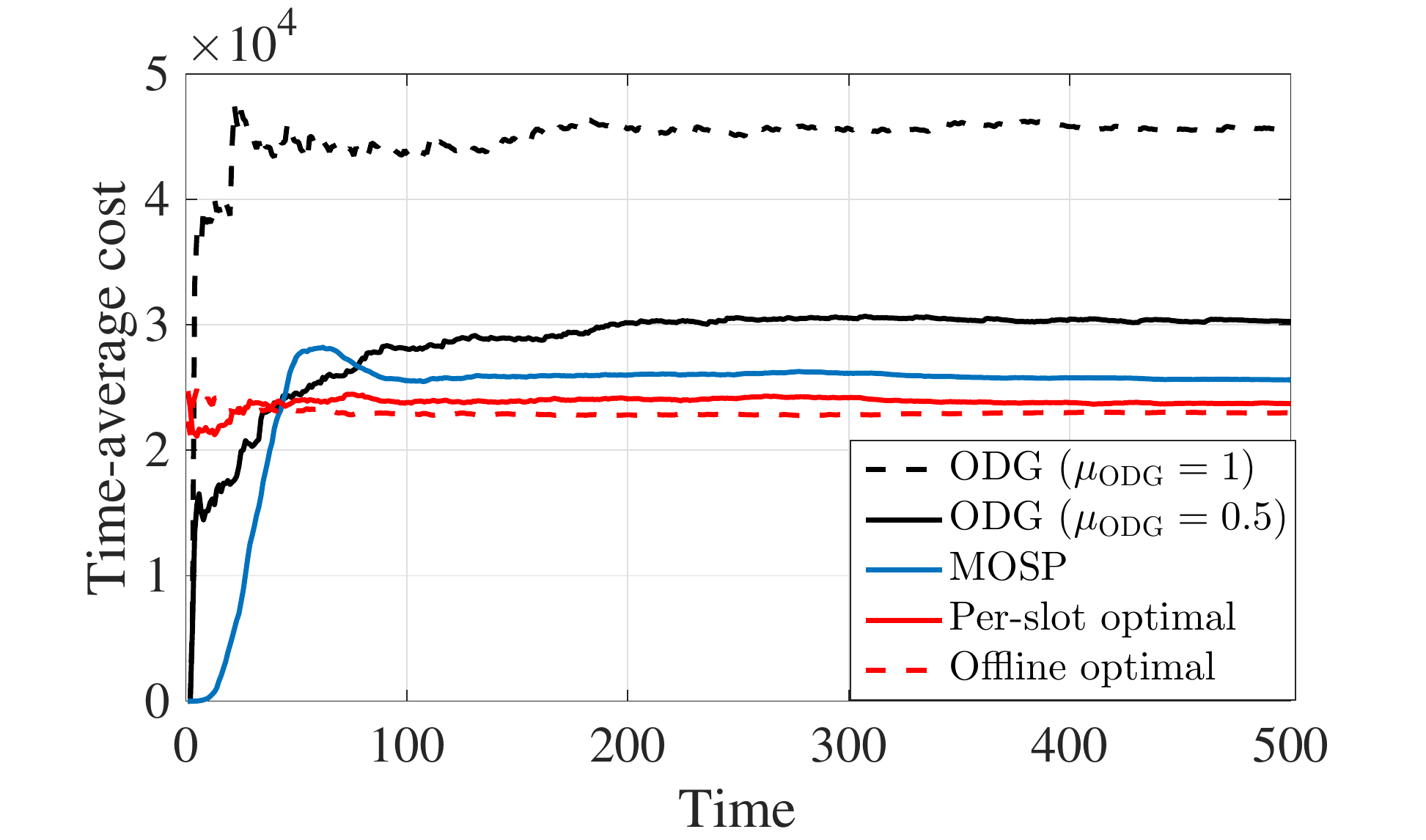}
\vspace{-0.5cm}
\caption{Time-average cost for Case 1.}
\label{Fig.cost1}
\end{figure}

\begin{figure}[t]
\centering
\vspace{-0.2cm}
\includegraphics[height=0.3\textwidth]{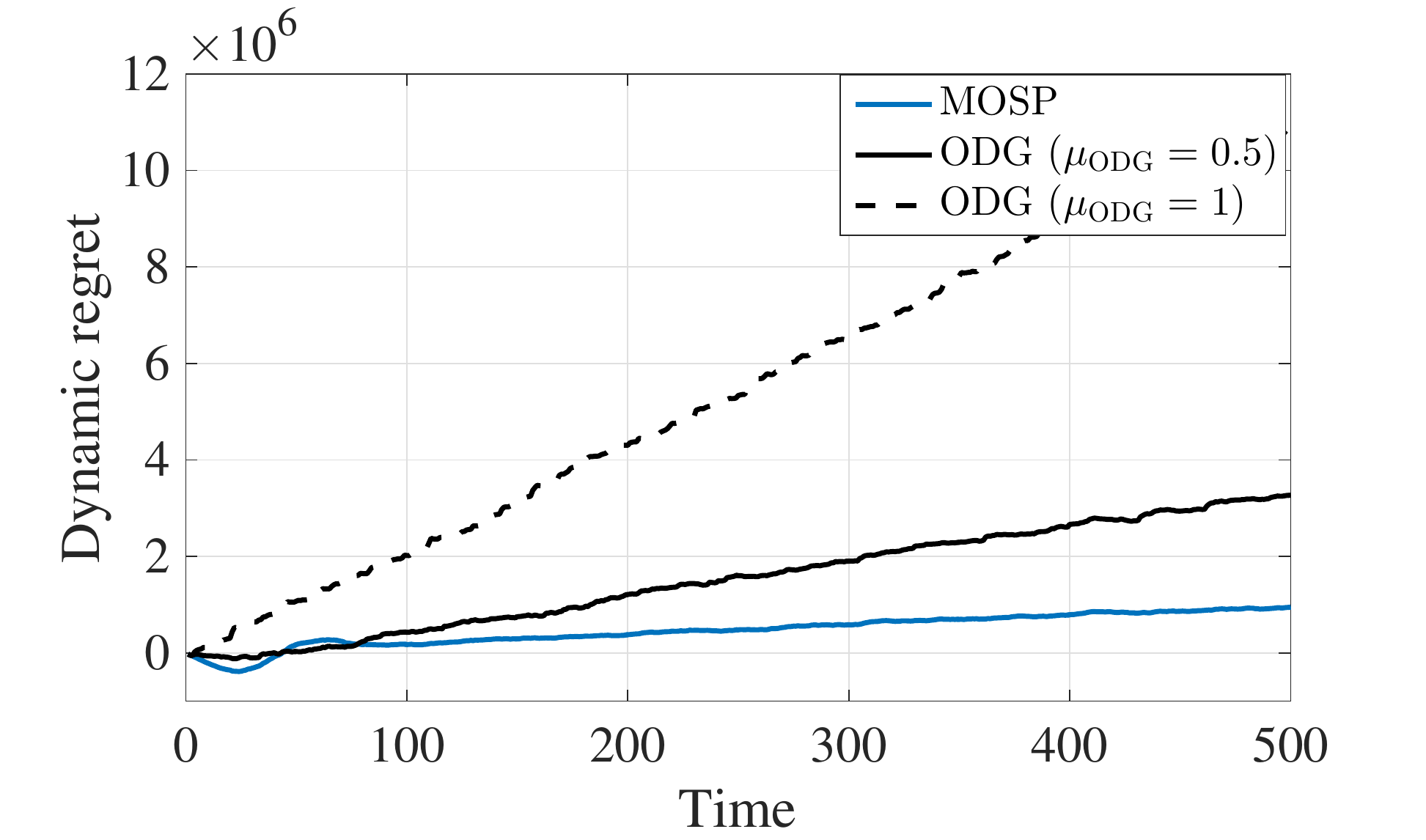}
\vspace{-0.5cm}
\caption{Dynamic regret for Case 1.}
\label{Fig.reg1}
\end{figure}

\begin{figure}[t]
\centering
\vspace{-0.2cm}
\includegraphics[height=0.3\textwidth]{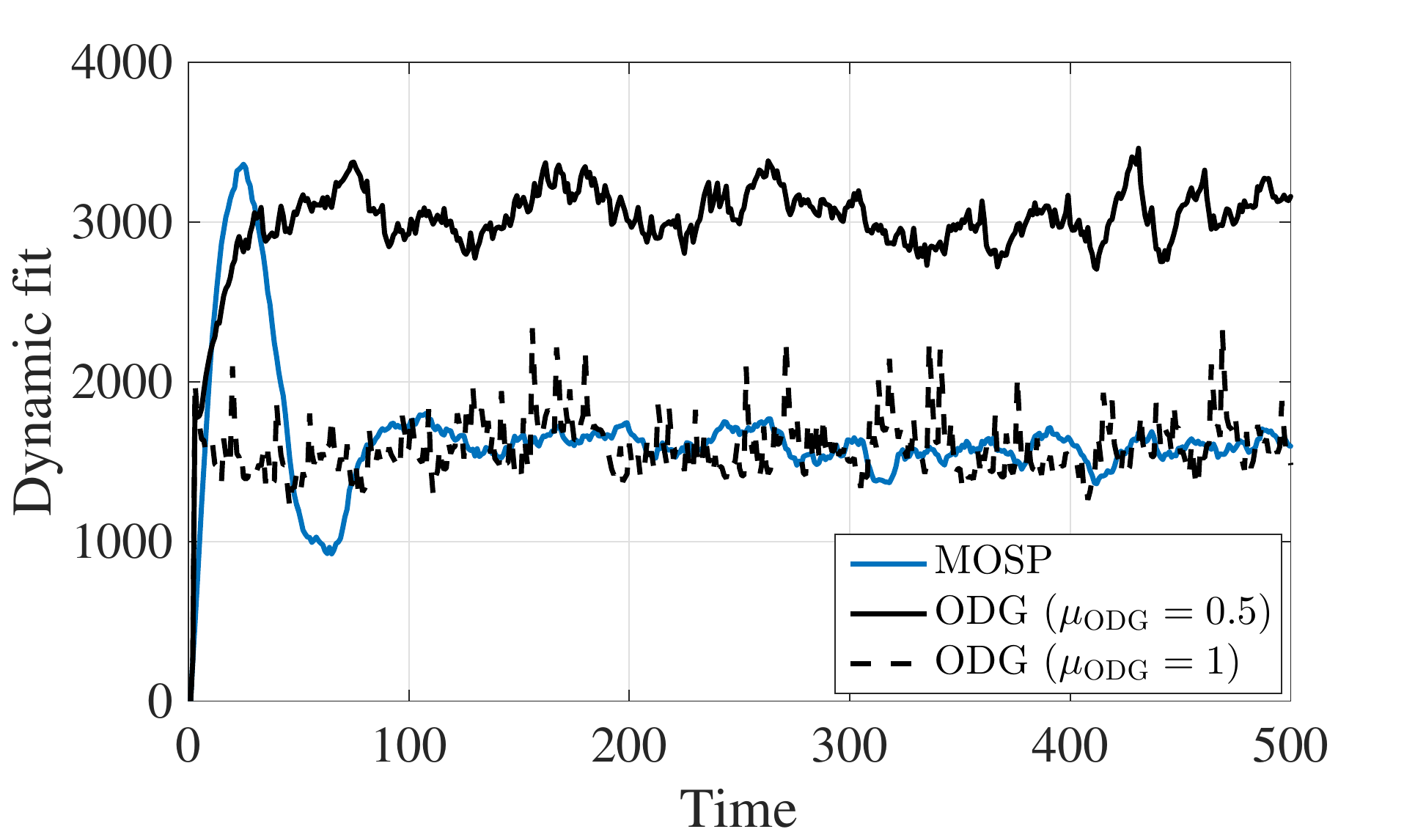}
\vspace{-0.5cm}
\caption{Dynamic fit for Case 1.}
\label{Fig.fit1}
\vspace{-0.2cm}
\end{figure}

\begin{figure}[t]
\centering
\vspace{-0.2cm}
\includegraphics[height=0.3\textwidth]{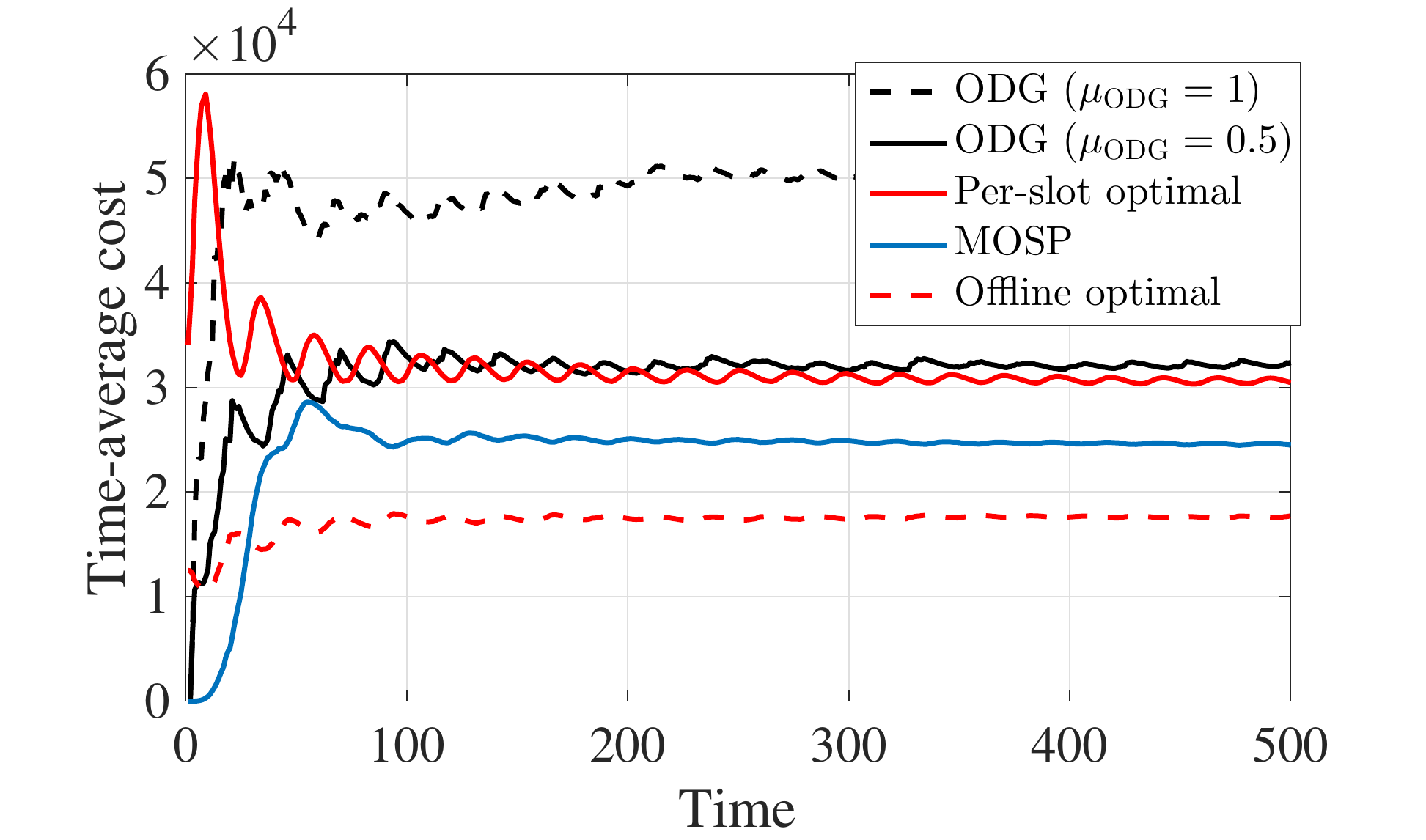}
\vspace{-0.5cm}
\caption{Time-average cost for Case 2.}
\label{Fig.cost2}
\vspace{-0.2cm}
\end{figure}

Figs. \ref{Fig.cost1}-\ref{Fig.fit1} show the test results for
Case 1 under i.i.d. costs and constraints. Clearly, MOSP in Fig.
\ref{Fig.cost1} converges to a smaller time-average cost than ODG
with the two stepsizes. The time-average cost of MOSP is
slightly higher than the per-slot optimal solution, as well as the
offline optimal solution with all information of the costs and
constraints available over horizon $T$. Fig. \ref{Fig.reg1} confirms the
conclusion made from Fig. \ref{Fig.cost1}, where the dynamic
regret (cf. \eqref{eq.dyn-reg}) of MOSP grows much slower than
that of ODG. Regarding the dynamic fit (cf. \eqref{eq.dyn-fit}),
Fig. \ref{Fig.fit1} demonstrates that ODG with $\mu_{\rm ODG}=1$
has a smaller fit than that of $\mu_{\rm ODG}=0.5$, and similar to
the dynamic fit of MOSP. According to the well-known trade-off
between cost (optimality) and delay (constraint violations) in
\cite{neely2010}, increasing $\mu_{\rm ODG}$ will improve
the dynamic fit of ODG but degrade its dynamic regret. Therefore,
MOSP is favorable in Case 1 since it has much smaller regret when
its dynamic fit is similar to that of ODG with $\mu_{\rm ODG}=1$.
It is worth mentioning that theoretically speaking, the dynamic
regret of MOSP may not be sub-linear in this i.i.d. case, since
the accumulated cost and constraint variation is not necessarily
small enough (cf. Theorem \ref{Them2}). However, MOSP is robust in this aspect at least for the numerical tests we carried.

\begin{figure}[t]
\centering
\vspace{-0.2cm}
\includegraphics[height=0.3\textwidth]{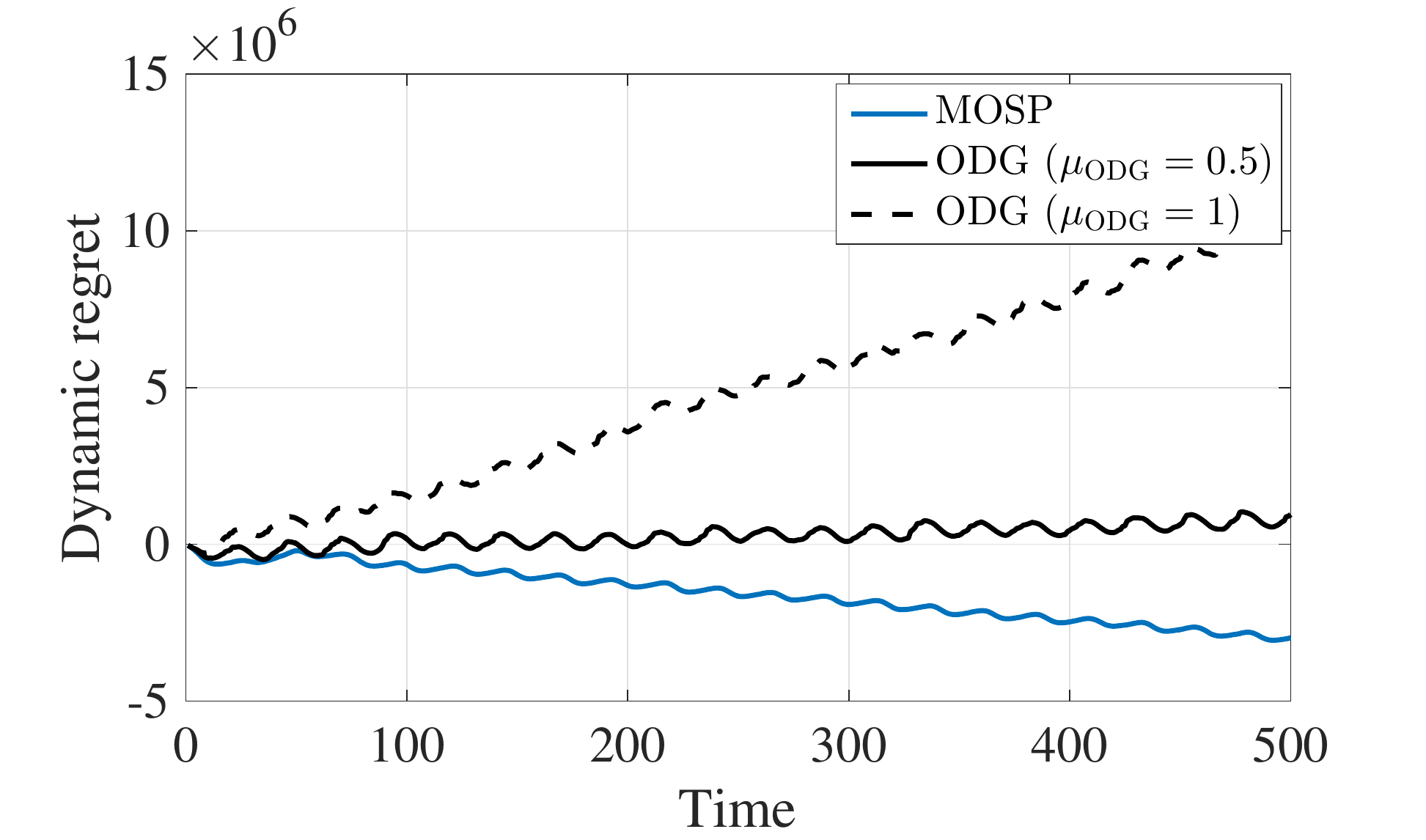}
\vspace{-0.5cm}
\caption{Dynamic regret for Case 2.}
\label{Fig.reg2}
\end{figure}

\begin{figure}[t]
\centering
\vspace{-0.2cm}
\includegraphics[height=0.3\textwidth]{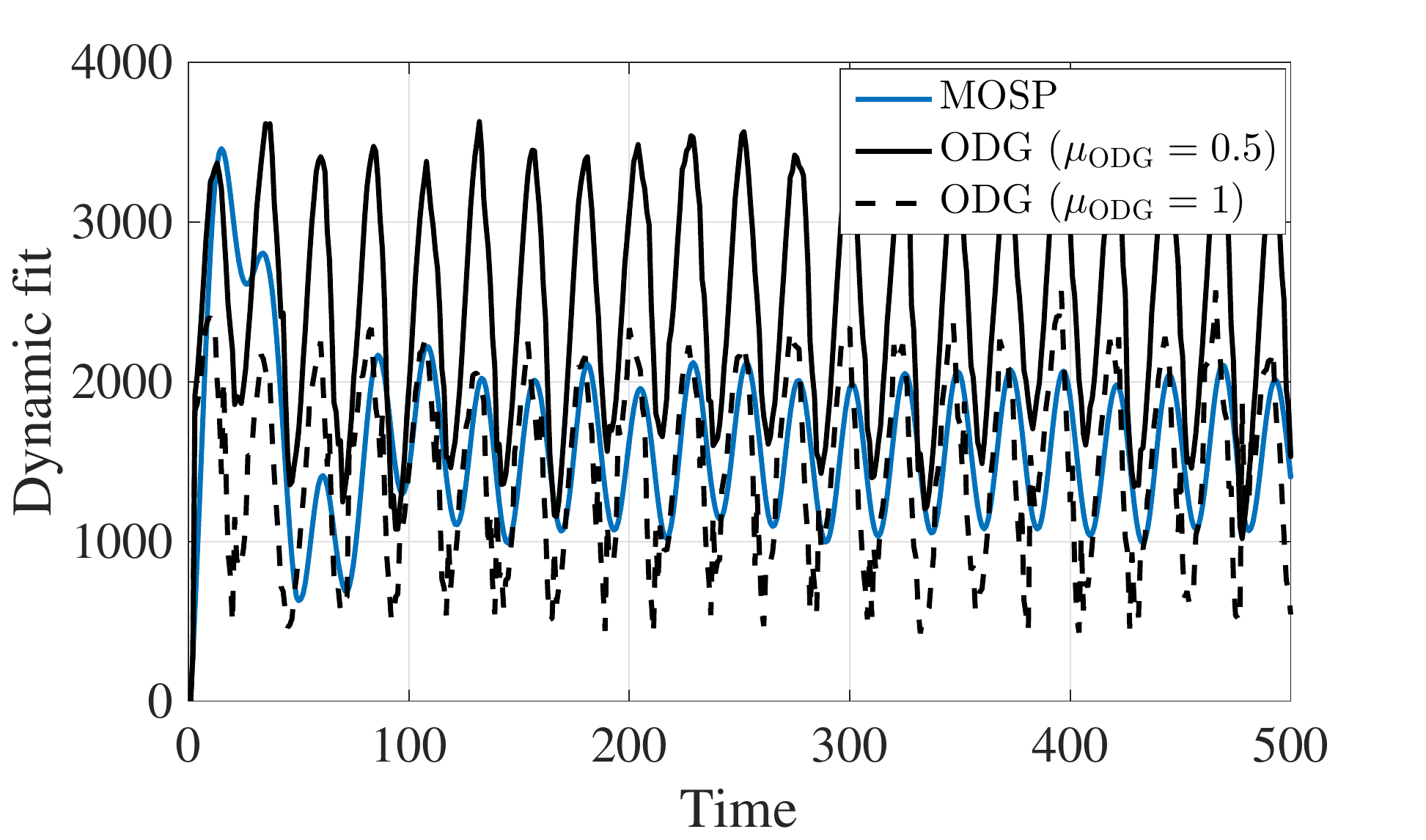}
\vspace{-0.5cm}
\caption{Dynamic fit for Case 2.}
\label{Fig.fit2}
\end{figure}

Simulation tests using non-stationary costs and constraints are
shown in Figs. \ref{Fig.cost2}-\ref{Fig.fit2}. Different from Case 1, the time-average cost of MOSP is not only smaller than
ODG, but also smaller than the per-slot optimum obtained via
\eqref{eq.slot-opt}; see Fig. \ref{Fig.cost2}. A similar
conclusion can be also drawn through the growths of dynamic
regret in Fig. \ref{Fig.reg2}.
From a high level, this is because the difference between the cost of the per-slot minimizers and that of the offline solutions is no longer small in the non-stationary case.
Regarding Fig. \ref{Fig.fit2}, both
ODG and MOSP have finite dynamic fits in the sense that the
accumulated constraint violations do not increase with time. The
dynamic fit of MOSP is much smaller than that of ODG with
$\mu_{\rm ODG}=0.5$, and comparable to that of ODG with $\mu_{\rm
ODG}=1$. Therefore, in this non-stationary case, MOSP also
significantly outperforms ODG in both dynamic regret and
fit.

\section{Concluding Remarks}

OCO with both adversarial costs and constraints has been studied
in this paper. Different from existing works, the focus is on a
setting where some of the constraints are revealed after taking
actions, they are tolerable to instantaneous violations, but must be
satisfied on average. Performance of the novel OCO algorithm is measured
by: i) the difference of its objective relative to the best
dynamic solution with one-slot-ahead information of the cost and
the constraint (dynamic regret); and, ii) its accumulated amount of
constraint violations (dynamic fit). It has been shown that the
proposed MOSP algorithm adapts to the considered OCO setting with
adversarial constraints.
Under standard assumptions, MOSP simultaneously
yields sub-linear dynamic regret and fit, if the accumulated
variations of the per-slot minimizers and adversarial constraints
are sub-linearly increasing with time. Algorithm design and
performance analysis in this novel OCO setting, under adversarial
constraints and with a dynamic benchmark, broaden the
applicability of OCO to a wider application regime, which includes dynamic network resource allocation and online demand response in smart grids. 
Numerical tests demonstrated that the proposed algorithm
outperforms state-of-the-art alternatives under different scenarios.

\section*{Acknowledgement}
The authors would like to thank Prof. Xin Wang for
helpful discussions and comments.


\appendix

Before proving Theorems \ref{Them1} and \ref{Them2}, we first
bound the variation of the dual variable for the MOSP recursion
\eqref{eq.primal}-\eqref{eq.dual}. With the dual drift defined as
$\Delta(\bm{\lambda}_t):=\left(\|\bm{\lambda}_{t+1}\|^2-\|\bm{\lambda}_t\|^2\right)/2$,
we have the following lemma.
\begin{lemma}\label{Lemma1}
Per slot $t$, the dual drift of the MOSP recursion
\eqref{eq.primal}-\eqref{eq.dual} is upper-bounded as
\begin{align}\label{eq.lemma1}
    \Delta(\bm{\lambda}_t)\leq \mu\bm{\lambda}_t^{\top}\mathbf{g}_t(\mathbf{x}_t)+\frac{\mu^2}{2}\|\mathbf{g}_t(\mathbf{x}_t)\|^2.
\end{align}
\end{lemma}
\begin{proof}
    Squaring the dual variable update \eqref{eq.dual}, we have
    \begin{align}\label{eq.dualdrift}
        \|\bm{\lambda}_{t+1}\|^2 & = \left\|\big[\bm{\lambda}_t+\mu \mathbf{g}_t(\mathbf{x}_t)\big]^{+}\right\|^2  \leq \left\|\bm{\lambda}_t+\mu \mathbf{g}_t(\mathbf{x}_t)\right\|^2 \nonumber \\
                                 & = \|\bm{\lambda}_t\|^2+2\mu\bm{\lambda}_t^{\top}\mathbf{g}_t(\mathbf{x}_t)+\mu^2\|\mathbf{g}_t(\mathbf{x}_t)\|^2.
    \end{align}
    And the proof is complete after rearranging terms and dividing both sides by 2.
\end{proof}

\subsection{Proof of Theorem \ref{Them1}}\label{app.thm1}
The proof follows the steps in \cite[Theorem 7]{yu2016}, but
generalizes the result from static regret with time-invariant
constraints to dynamic regret with time-varying and long-term
constraints. Recall that the primal iterate $\mathbf{x}_{t+1}$ is
the optimal solution to the following optimization problem (cf.
\eqref{eq.primal})
\begin{equation*}
	    \min_{\mathbf{x}\in {\cal X}}\; \nabla^{\top}
    f_{t}(\mathbf{x}_t)\!\left(\mathbf{x}-\mathbf{x}_t\right)+\bm{\lambda}_{t+1}^{\top}\mathbf{g}_t(\mathbf{x})+\frac{1}{2\alpha}\|\mathbf{x}-\mathbf{x}_t\|^2.
\end{equation*}
Then for any interior point $\tilde{\mathbf{x}}_t\in {\cal X}$ in Assumption \ref{ass.3}, it follows that
\begin{align}\label{eq.29}
    &\nabla^{\top} f_{t}(\mathbf{x}_t)\!\left(\mathbf{x}_{t+1}\!-\!\mathbf{x}_t\right)+\bm{\lambda}_{t+1}^{\top}\mathbf{g}_t(\mathbf{x}_{t+1})\!+\!\frac{1}{2\alpha}\|\mathbf{x}_{t+1}\!-\!\mathbf{x}_t\|^2\nonumber\\
\leq &  \nabla^{\top} f_{t}(\mathbf{x}_t)\!\left(\tilde{\mathbf{x}}_t\!-\!\mathbf{x}_t\right)+\bm{\lambda}_{t+1}^{\top}\mathbf{g}_t(\tilde{\mathbf{x}}_t)+\frac{1}{2\alpha}\|\tilde{\mathbf{x}}_t\!-\!\mathbf{x}_t\|^2\nonumber\\
\stackrel{(a)}{\leq} & \nabla^{\top} f_{t}(\mathbf{x}_t)\!\left(\tilde{\mathbf{x}}_t\!-\!\mathbf{x}_t\right)-\epsilon\bm{\lambda}_{t+1}^{\top}\mathbf{1}+\frac{1}{2\alpha}\|\tilde{\mathbf{x}}_t\!-\!\mathbf{x}_t\|^2\nonumber\\
\stackrel{(b)}{\leq} &\nabla^{\top} f_{t}(\mathbf{x}_t)\!\left(\tilde{\mathbf{x}}_t\!-\!\mathbf{x}_t\right)-\epsilon\|\bm{\lambda}_{t+1}\|+\frac{1}{2\alpha}\|\tilde{\mathbf{x}}_t\!-\!\mathbf{x}_t\|^2
\end{align}
where (a) follows by choosing
$\tilde{\mathbf{x}}_t$ such that
$\mathbf{g}_t(\tilde{\mathbf{x}}_t) \leq -\epsilon\mathbf{1}$ and recalling the
non-negativity of $\bm{\lambda}_{t+1}$;
inequality (b) is because $\|\bm{\lambda}_{t+1}\|\leq
\bm{\lambda}_{t+1}^{\top}\mathbf{1}$ holds for any
non-negative vector $\bm{\lambda}_{t+1}$.

Rearranging terms in \eqref{eq.29}, it follows that
\begin{align}\label{eq.crossbound}
   & \bm{\lambda}_{t+1}^{\top}\mathbf{g}_t(\mathbf{x}_{t+1})\leq\nabla^{\top} f_{t}(\mathbf{x}_t)\!\left(\tilde{\mathbf{x}}_t-\mathbf{x}_t\right)\!-\!\nabla^{\top} f_{t}(\mathbf{x}_t)\!\left(\mathbf{x}_{t+1}-\mathbf{x}_t\right)\nonumber\\
   &\qquad\qquad\qquad\quad-\epsilon\|\bm{\lambda}_{t+1}\|\!+\!\frac{1}{2\alpha}\|\tilde{\mathbf{x}}_t-\mathbf{x}_t\|^2\!-\!\frac{1}{2\alpha}\|\mathbf{x}_{t+1}-\mathbf{x}_t\|^2\nonumber\\
   & \stackrel{(c)}{\leq}\! \nabla^{\top} f_{t}(\mathbf{x}_t)\!\left(\tilde{\mathbf{x}}_t-\mathbf{x}_t\right)\!-\!\nabla^{\top} f_{t}(\mathbf{x}_t)\!\left(\mathbf{x}_{t+1}-\mathbf{x}_t\right)\!-\!\epsilon\|\bm{\lambda}_{t+1}\|\!+\!\frac{R^2}{2\alpha}\nonumber\\
   &\stackrel{(d)}{\leq}\!\! \|\nabla f_{t}(\mathbf{x}_t)\|\|\tilde{\mathbf{x}}_t\!-\!\mathbf{x}_t\|\!+\!\|\nabla f_{t}(\mathbf{x}_t)\|\|\mathbf{x}_{t+1}\!-\!\mathbf{x}_t\|\!-\!\epsilon\|\bm{\lambda}_{t+1}\|\!+\!\frac{R^2}{2\alpha}\nonumber\\
   & \stackrel{(e)}{\leq} 2GR-\epsilon\|\bm{\lambda}_{t+1}\|+\frac{R^2}{2\alpha}
\end{align}
where (c) holds since ${\cal X}$ confines $\|\tilde{\mathbf{x}}_t-\mathbf{x}_t\|^2\leq R^2$ and
$\|\mathbf{x}_{t+1}-\mathbf{x}_t\|^2\geq 0$; (d)
uses the Cauchy-Schwartz inequality twice; (e) leverages the bounds in
Assumption \ref{ass.2}, namely, $\|\nabla
f_{t}(\mathbf{x}_t)\|\leq G$,
$\|\tilde{\mathbf{x}}_t-\mathbf{x}_t\|\leq R$, and
$\|\mathbf{x}_{t+1}\!-\!\mathbf{x}_t\|\leq R$.

Plugging \eqref{eq.crossbound} into \eqref{eq.dualdrift} in Lemma
\ref{Lemma1}, we have
\begin{align}\label{eq.drift1}
    &\Delta(\bm{\lambda}_{t+1}) \leq \mu\bm{\lambda}_{t+1}^{\top}\mathbf{g}_{t+1}(\mathbf{x}_{t+1})+\frac{\mu^2}{2}\|\mathbf{g}_{t+1}(\mathbf{x}_{t+1})\|^2\nonumber\\
     &\stackrel{(f)}{\leq} \mu\bm{\lambda}_{t+1}^{\top}\big(\mathbf{g}_{t+1}(\mathbf{x}_{t+1})-\mathbf{g}_t(\mathbf{x}_{t+1})\big)-\epsilon\mu\|\bm{\lambda}_{t+1}\|\nonumber\\
     &\qquad\qquad\qquad\qquad\qquad\qquad\qquad~~+2\mu GR+\frac{\mu R^2}{2\alpha}+\frac{\mu^2 M^2}{2}\nonumber\\
     &\stackrel{(g)}{\leq}\mu\bm{\lambda}_{t+1}^{\top}\big[\mathbf{g}_{t+1}(\mathbf{x}_{t+1})-\mathbf{g}_t(\mathbf{x}_{t+1})\big]^+\!\!\!-\epsilon\mu\|\bm{\lambda}_{t+1}\|\nonumber\\
     &\qquad\qquad\qquad\qquad\qquad\qquad\qquad~~+2\mu GR+\frac{\mu R^2}{2\alpha}+\frac{\mu^2 M^2}{2}\nonumber\\
     &\stackrel{(h)}{\leq}\mu\bar{V}(\mathbf{g})\|\bm{\lambda}_{t+1}\|\!-\!\epsilon\mu\|\bm{\lambda}_{t+1}\|\!+\!2\mu GR\!+\!\frac{\mu R^2}{2\alpha}\!+\!\frac{\mu^2 M^2}{2}\!\!\!
\end{align}
where (f) uses the upper bound in Assumption \ref{ass.1} such that
$\|\mathbf{g}_{t+1}(\mathbf{x}_{t+1})\|\leq M$, (g) holds since
$\bm{\lambda}_{t+1}\geq \mathbf{0}$, and (h) follows from the
Cauchy-Schwartz inequality and the definition of the maximum variation
$\bar{V}(\mathbf{g})$ in \eqref{eq.maxgt}.

We prove the dual upper bound \eqref{eq.upper-dual} by
contradiction. Without loss of generality, suppose that $t+2$ is the first time that
\eqref{eq.upper-dual} does not hold. Therefore, we have
\begin{subequations}
    \begin{equation}\label{eq.temp-001}
    \|\bm{\lambda}_{t+1}\|\leq \|\bar{\bm{\lambda}}\| = \mu M+\frac{2 GR+{R^2}/(2\alpha)+(\mu M^2)/{2}}{\epsilon-\bar{V}(\mathbf{g})}
\end{equation}
and correspondingly
\begin{equation}\label{eq.temp-002}
    \|\bm{\lambda}_{t+2}\|> \|\bar{\bm{\lambda}}\| = \mu M+\frac{2 GR+{R^2}/(2\alpha)+(\mu M^2)/{2}}{\epsilon-\bar{V}(\mathbf{g})}.
\end{equation}
\end{subequations}
In this case, it follows that
\begin{align}\label{eq.temp-003}
    \|\bm{\lambda}_{t+1}\| & \geq \|\bm{\lambda}_{t+2}\| - \|\bm{\lambda}_{t+2} - \bm{\lambda}_{t+1}\| \nonumber \\
                           & = \|\bm{\lambda}_{t+2}\| - \|[\bm{\lambda}_{t+1} + \mu \mathbf{g}_{t+1}(\mathbf{x}_{t+1})]^+ - \bm{\lambda}_{t+1}\| \nonumber \\
                           &\stackrel{(i)}{\geq} \|\bm{\lambda}_{t+2}\| - \|\mu \mathbf{g}_{t+1}(\mathbf{x}_{t+1})\| \nonumber \\
                           & \stackrel{(j)}{>} \frac{2 GR+{R^2}/(2\alpha)+(\mu M^2)/{2}}{\epsilon-\bar{V}(\mathbf{g})}
\end{align}
where (i) is due to the non-expansive property of the projection operator, and inequality (j) uses \eqref{eq.temp-002} and $\|\mathbf{g}_{t+1}(\mathbf{x}_{t+1})\|\leq
M$ in Assumption \ref{ass.1}.
However, since $\epsilon> \bar{V}(\mathbf{g})$, \eqref{eq.drift1} implies that we have $\Delta(\bm{\lambda}_{t+1})<0$ if
\eqref{eq.temp-003} holds. By definition of the dual drift,
$\Delta(\bm{\lambda}_{t+1})<0$ implies that
$\|\bm{\lambda}_{t+2}\|<\|\bm{\lambda}_{t+1}\|$, which contradicts
\eqref{eq.temp-001} and \eqref{eq.temp-002}. In addition,
observe that the dual variable is initialized by
$\bm{\lambda}_1 = \mathbf{0}$, and consequently
$\|\bm{\lambda}_2\| \leq \mu M$. Therefore, for every $t$,
we have that $\|\bm{\lambda}_t\|\leq \|\bar{\bm{\lambda}}\|$ holds.

Using the dual recursion in \eqref{eq.dual}, it follows that $  \bm{\lambda}_{T+1}\geq \bm{\lambda}_T+\mu \mathbf{g}_T(\mathbf{x}_T)\geq \bm{\lambda}_1+\sum_{t=1}^T\mu \mathbf{g}_t(\mathbf{x}_t)$.
Rearranging terms, we have
\begin{equation}\label{eq.36}
    \sum_{t=1}^T\mathbf{g}_t(\mathbf{x}_t) \leq \frac{\bm{\lambda}_{T+1}}{\mu}-\frac{\bm{\lambda}_1}{\mu}\leq \frac{\bm{\lambda}_{T+1}}{\mu}.
\end{equation}
With $\bm{\lambda}_{T+1}\geq \mathbf{0}$, \eqref{eq.36} implies
that
$\left[\sum_{t=1}^T\mathbf{g}_t(\mathbf{x}_t)\right]^+\!\!\leq
{\bm{\lambda}_{T+1}}/{\mu}$, which completes the proof by taking
norms on both sides and using the dual upper bound
\eqref{eq.upper-dual}.

\subsection{Proof of Theorem \ref{Them2}}\label{app.thm2}
    Per slot $t$, the primal update $\mathbf{x}_{t+1}$ is the minimizer of the optimization problem in \eqref{eq.primal}; hence,
    \begin{align}\label{eq.33}
        \!\!\! &\nabla^{\top} f_t(\mathbf{x}_t)\left(\mathbf{x}_{t+1}\!-\!\mathbf{x}_t\right)\!+\!\bm{\lambda}_{t+1}^{\top}\mathbf{g}_t(\mathbf{x}_{t+1})\!+\!\frac{\|\mathbf{x}_{t+1}\!-\!\mathbf{x}_t\|^2}{2\alpha}\\
       \!\!\! \stackrel{(a)}{\leq} &\nabla^{\top}\!\! f_t(\mathbf{x}_t)(\mathbf{x}^*_t\!-\!\mathbf{x}_t)\!+\!\bm{\lambda}_{t+1}^{\top}\mathbf{g}_t(\mathbf{x}_t^*)\!+\!\frac{\|\mathbf{x}^*_t\!-\!\mathbf{x}_t\|^2}{2\alpha}\!-\!\frac{\|\mathbf{x}_{t+1}\!-\!\mathbf{x}^*_t\|^2}{2\alpha}\nonumber
    \end{align}
    where (a) uses the strong convexity of the objective in \eqref{eq.primal}; see also \cite[Corollary 1]{yu2016}.
    Adding $f_t(\mathbf{x}_t)$ in \eqref{eq.33} yields
    \begin{align}
        &f_t(\mathbf{x}_t)\!+\!\nabla^{\top}\! f_t(\mathbf{x}_t)\left(\mathbf{x}_{t+1}\!-\!\mathbf{x}_t\right)\!+\!\bm{\lambda}_{t+1}^{\top}\mathbf{g}_t(\mathbf{x}_{t+1})\!+\!\frac{\|\mathbf{x}_{t+1}\!-\!\mathbf{x}_t\|^2}{2\alpha}\nonumber\\
        \leq &f_t(\mathbf{x}_t)\!+\!\nabla^{\top} f_t(\mathbf{x}_t)\left(\mathbf{x}^*_t\!-\!\mathbf{x}_t\right)\!+\!\bm{\lambda}_{t+1}^{\top}\mathbf{g}_t(\mathbf{x}_t^*)\nonumber\\
        &\qquad\qquad\qquad\qquad\qquad\qquad~\,+\frac{\|\mathbf{x}^*_t\!-\!\mathbf{x}_t\|^2}{2\alpha}-\frac{\|\mathbf{x}^*_t\!-\!\mathbf{x}_{t+1}\|^2}{2\alpha}\nonumber\\
        \stackrel{(b)}{\leq} &f_t(\mathbf{x}^*_t)+\bm{\lambda}_{t+1}^{\top}\mathbf{g}_t(\mathbf{x}_t^*)\!+\!\frac{\|\mathbf{x}^*_t\!-\!\mathbf{x}_t\|^2}{2\alpha}\!-\!\frac{\|\mathbf{x}^*_t\!-\!\mathbf{x}_{t+1}\|^2}{2\alpha}\nonumber\\
        \stackrel{(c)}{\leq}&f_t(\mathbf{x}^*_t)+\frac{\|\mathbf{x}^*_t\!-\!\mathbf{x}_t\|^2}{2\alpha}-\frac{\|\mathbf{x}^*_t\!-\!\mathbf{x}_{t+1}\|^2}{2\alpha}\label{eq.bd0}
    \end{align}
where (b) is due to the convexity of $f_t(\mathbf{x})$, and (c)
comes from the fact that $\bm{\lambda}_{t+1}\geq \mathbf{0}$ and
the per-slot optimal solution $\mathbf{x}^*_t$ is feasible (i.e.,
$\mathbf{g}_t(\mathbf{x}_t^*)\leq \mathbf{0}$) such that
$\bm{\lambda}_{t+1}^{\top}\mathbf{g}_t(\mathbf{x}_t^*)\leq 0$.

    Next, we bound the term $\nabla^{\top} f_t(\mathbf{x}_t)\left(\mathbf{x}_{t+1}\!-\!\mathbf{x}_t\right)$ by
    \begin{align}\label{eq.bd1}
    &-\nabla^{\top} f_t(\mathbf{x}_t)\left(\mathbf{x}_{t+1}\!-\!\mathbf{x}_t\right)\leq \|\nabla f_t(\mathbf{x}_t)\|\|\mathbf{x}_{t+1}-\mathbf{x}_t\|\\
    \leq &\frac{\|\nabla f_t(\mathbf{x}_t)\|^2}{2\eta}+\frac{\eta}{2}\|\mathbf{x}_{t+1}\!-\!\mathbf{x}_t\|^2 \stackrel{(d)}{\leq}\frac{G^2}{2\eta}+\frac{\eta}{2}\|\mathbf{x}_{t+1}-\mathbf{x}_t\|^2\nonumber
    \end{align}
where $\eta$ is an arbitrary positive constant, and (d) is from
the bound of gradients in Assumption \ref{ass.1}. Plugging
\eqref{eq.bd1} into \eqref{eq.bd0}, we have
    \begin{align} \label{eq.temp-004}
        &f_t(\mathbf{x}_t)+\bm{\lambda}_{t+1}^{\top}\mathbf{g}_t(\mathbf{x}_{t+1}) \leq  f_t(\mathbf{x}^*_t)+\Big(\frac{\eta}{2}-\frac{1}{2\alpha}\Big)\|\mathbf{x}_{t+1}\!-\!\mathbf{x}_t\|^2\nonumber\\
        &\qquad~\,+\frac{1}{2\alpha}\Big(\|\mathbf{x}^*_t\!-\!\mathbf{x}_t\|^2\!-\!\|\mathbf{x}^*_t\!-\!\mathbf{x}_{t+1}\|^2\Big)+\frac{G^2}{2\eta}\nonumber\\
    \stackrel{(e)}{=}&f_t(\mathbf{x}^*_t)\!+\!\frac{1}{2\alpha}\Big(\|\mathbf{x}^*_t\!-\!\mathbf{x}_t\|^2\!-\!\|\mathbf{x}^*_t\!-\!\mathbf{x}_{t+1}\|^2\Big)\!+\!\frac{\alpha G^2}{2}
    \end{align}
where equality (e) follows by choosing $\eta=1/\alpha$ to obtain $\eta/2-1/(2\alpha)=0$.

Using the dual drift bound \eqref{eq.lemma1} in Lemma
\ref{Lemma1} again, we have
 \begin{align}\label{eq.55}
        &\Delta(\bm{\lambda}_{t+1})/\mu+f_t(\mathbf{x}_t)\leq f_t(\mathbf{x}_t)+\bm{\lambda}_{t+1}^{\top}\mathbf{g}_t(\mathbf{x}_{t+1})\nonumber\\
        &\qquad+\bm{\lambda}_{t+1}^{\top}\mathbf{g}_{t+1}(\mathbf{x}_{t+1})-\bm{\lambda}_{t+1}^{\top}\mathbf{g}_t(\mathbf{x}_{t+1})+\frac{\mu}{2}\|\mathbf{g}_{t+1}(\mathbf{x}_{t+1})\|^2\nonumber\\
    \stackrel{(f)}{\leq} &f_t(\mathbf{x}^*_t)\!+\!\frac{1}{2\alpha}\Big(\|\mathbf{x}^*_t\!-\!\mathbf{x}_t\|^2\!-\!\|\mathbf{x}^*_t\!-\!\mathbf{x}_{t+1}\|^2\Big)\nonumber\\
    &+\!\bm{\lambda}_{t+1}^{\top}(\mathbf{g}_{t+1}(\mathbf{x}_{t+1})-\mathbf{g}_t(\mathbf{x}_{t+1}))\!+\!\frac{\mu\|\mathbf{g}_{t+1}(\mathbf{x}_{t+1})\|^2}{2}+\frac{\alpha G^2}{2}\nonumber\\
    \stackrel{(g)}{\leq}&f_t(\mathbf{x}^*_t)\!+\!\frac{1}{2\alpha}\Big(\|\mathbf{x}^*_t\!-\!\mathbf{x}_t\|^2\!-\!\|\mathbf{x}^*_t\!-\!\mathbf{x}_{t+1}\|^2\Big)\nonumber\\
    &\qquad\qquad+\bm{\lambda}_{t+1}^{\top}\left[\mathbf{g}_{t+1}(\mathbf{x}_{t+1})-\mathbf{g}_t(\mathbf{x}_{t+1})\right]^++\!\frac{\mu M^2}{2}\!+\frac{\alpha G^2}{2}\nonumber\\
    \stackrel{(h)}{\leq}&f_t(\mathbf{x}^*_t)\!+\!\frac{1}{2\alpha}\Big(\|\mathbf{x}^*_t\!-\!\mathbf{x}_t\|^2\!-\!\|\mathbf{x}^*_t\!-\!\mathbf{x}_{t+1}\|^2\Big)\!+\!\|\bm{\lambda}_{t+1}\|V(\mathbf{g}_t)\nonumber\\
    &\qquad\qquad+\frac{\mu M^2}{2}\!+\frac{\alpha G^2}{2}
    \end{align}
    where (f) follows from \eqref{eq.temp-004}; (g) uses non-negativity of $\bm{\lambda}_{t+1}$ and the gradient upper bound $\|\mathbf{g}_{t+1}(\mathbf{x})\|\leq M,\forall \mathbf{x}\in{\cal X}$; and (h) follows from the Cauchy-Schwartz inequality and the definition of the constraint variation $V(\mathbf{g}_t)$ in \eqref{eq.maxgt}.

   By interpolating intermediate terms in $\|\mathbf{x}^*_t\!-\!\mathbf{x}_t\|^2\!-\!\|\mathbf{x}^*_t\!-\!\mathbf{x}_{t+1}\|^2$, we have that
   \begin{align}\label{eq.56}
        &\|\mathbf{x}^*_t\!-\!\mathbf{x}_t\|^2\!-\!\|\mathbf{x}^*_t\!-\!\mathbf{x}_{t+1}\|^2\nonumber\\
       = &\|\mathbf{x}^*_t\!-\!\mathbf{x}_t\|^2\!-\!\|\mathbf{x}_t-\mathbf{x}^*_{t-1}\|^2+\|\mathbf{x}_t\!-\!\mathbf{x}^*_{t-1}\|^2\!-\!\|\mathbf{x}^*_t\!-\!\mathbf{x}_{t+1}\|^2\nonumber\\
        = &\|\mathbf{x}^*_t\!-\!\mathbf{x}^*_{t-1}\|\|\mathbf{x}^*_t-2\mathbf{x}_t+\mathbf{x}^*_{t-1}\|+\|\mathbf{x}_t\!-\!\mathbf{x}^*_{t-1}\|^2\!-\!\|\mathbf{x}^*_t\!-\!\mathbf{x}_{t+1}\|^2\nonumber\\
       \stackrel{(i)}{\leq} & 2R\|\mathbf{x}^*_t-\mathbf{x}^*_{t-1}\|+\|\mathbf{x}_t\!-\!\mathbf{x}^*_{t-1}\|^2\!-\!\|\mathbf{x}^*_t\!-\!\mathbf{x}_{t+1}\|^2
   \end{align}
   where (i) follows from the radius of ${\cal X}$ in Assumption \ref{ass.2} such that $\|\mathbf{x}^*_t-2\mathbf{x}_t+\mathbf{x}^*_{t-1}\|\leq \|\mathbf{x}^*_t-\mathbf{x}_t\|+\|\mathbf{x}_t-\mathbf{x}^*_{t-1}\|\leq 2R$.
Plugging \eqref{eq.56} into \eqref{eq.55}, it readily leads to
\begin{align}\label{eq.37}
    &\Delta(\bm{\lambda}_{t+1})/\mu+f_t(\mathbf{x}_t) \leq f_t(\mathbf{x}^*_t)\!+\!\|\bm{\lambda}_{t+1}\|V(\mathbf{g}_t)\!+\!\frac{\mu M^2}{2}\!+\frac{\alpha G^2}{2}
    \nonumber\\
    &+\frac{1}{2\alpha}\Big(2R\|\mathbf{x}^*_t\!-\!\mathbf{x}^*_{t-1}\|\!+\|\mathbf{x}_t\!-\!\mathbf{x}^*_{t-1}\|^2\!-\!\|\mathbf{x}^*_t\!-\!\mathbf{x}_{t+1}\|^2\Big).
\end{align}

Summing up \eqref{eq.37} over $t=1,2,\ldots,T$, we find
 \begin{align}
        &\sum_{t=1}^T\Delta(\bm{\lambda}_{t+1})/\mu+\sum_{t=1}^Tf_t(\mathbf{x}_t)\nonumber\\
    \leq &\sum_{t=1}^Tf_t(\mathbf{x}^*_t)\!+\!\frac{1}{2\alpha}\sum_{t=1}^T\left(\|\mathbf{x}_t\!-\!\mathbf{x}^*_{t-1}\|^2\!-\!\|\mathbf{x}^*_t\!-\!\mathbf{x}_{t+1}\|^2\right)\nonumber\\
    &\quad+\frac{RV(\{\mathbf{x}_t^*\}_{t=1}^T)}{\alpha}\!+\!\sum_{t=1}^T\|\bm{\lambda}_{t+1}\|V(\mathbf{g}_t)+\!\frac{\mu M^2T}{2}\!+\!\frac{\alpha G^2 T}{2}\nonumber\\
    \stackrel{(j)}{\leq}&\sum_{t=1}^Tf_t(\mathbf{x}^*_t)\!+\!\frac{1}{2\alpha}\left(\|\mathbf{x}_1\!-\!\mathbf{x}^*_0\|^2\!-\!\|\mathbf{x}^*_T\!-\!\mathbf{x}_{T+1}\|^2\right)+\!\frac{RV(\{\mathbf{x}_t^*\}_{t=1}^T)}{\alpha}\!\nonumber\\
    &\qquad\qquad~+\|\bar{\bm{\lambda}}\|\sum_{t=1}^T V(\mathbf{g}_t)\!+\!\frac{\mu M^2T}{2}\!+\!\frac{\alpha G^2T}{2}\nonumber\\
    \stackrel{(k)}{\leq} &\sum_{t=1}^Tf_t(\mathbf{x}^*_t)\!+\!\frac{1}{2\alpha}\left(\|\mathbf{x}_1\!-\!\mathbf{x}^*_0\|^2\right)\!+\!\frac{RV(\{\mathbf{x}_t^*\}_{t=1}^T)}{\alpha}\nonumber\\
    &\qquad\qquad~+\|\bar{\bm{\lambda}}\|V(\{\mathbf{g}_t\}_{t=1}^T)+\frac{\mu M^2T}{2}\!+\!\frac{\alpha G^2T}{2}
    \end{align}
    where (j) uses the upper bound of $\|\bm{\lambda}_t\|$ in \eqref{eq.upper-dual} that we define as $\|\bar{\bm{\lambda}}\|$, and (k) follows from the definition of accumulated variations $V(\{\mathbf{g}_t\}_{t=1}^T)$ in \eqref{eq.var-gt}. The definition of dynamic regret in \eqref{eq.dyn-reg} finally implies that
     \begin{align}
        {\rm Reg}^{\rm d}_T\leq &\frac{RV(\{\mathbf{x}_t^*\}_{t=1}^T)}{\alpha}+\!\frac{\|\mathbf{x}_1\!-\!\mathbf{x}^*_0\|^2}{2\alpha}\!+\!\|\bar{\bm{\lambda}}\|V(\{\mathbf{g}_t\}_{t=1}^T)\nonumber\\
        &+\!\frac{\mu M^2T}{2}\!+\!\frac{\alpha G^2T}{2}-\sum_{t=1}^T\frac{\Delta(\bm{\lambda}_{t+1})}{\mu}\nonumber\\
        = &\frac{RV(\{\mathbf{x}_t^*\}_{t=1}^T)}{\alpha}+\!\frac{\|\mathbf{x}_1\!-\!\mathbf{x}^*_0\|^2}{2\alpha}\!+\!\|\bar{\bm{\lambda}}\|V(\{\mathbf{g}_t\}_{t=1}^T)\nonumber\\
        &+\!\frac{\mu M^2T}{2}\!+\!\frac{\alpha G^2T}{2}-\frac{\|\bm{\lambda}_{T+2}\|^2}{2\mu}+\frac{\|\bm{\lambda}_2\|^2}{2\mu}\nonumber\\
        \stackrel{(l)}{\leq}&\frac{RV(\{\mathbf{x}_t^*\}_{t=1}^T)}{\alpha}+\!\frac{R^2}{2\alpha}\!+\!\|\bar{\bm{\lambda}}\|V(\{\mathbf{g}_t\}_{t=1}^T)\nonumber\\
        &+\!\frac{\mu M^2T}{2}\!+\!\frac{\alpha G^2T}{2}+\frac{\mu
        M^2}{2}
    \end{align}
    where ($l$) follows since: i) $\|\mathbf{x}_1-\mathbf{x}^*_0\| \leq R$ due to the compactness of $\mathcal{X}$; ii) $\|\bm{\lambda}_{T+2}\|^2\geq 0$; and, iii) $\|\bm{\lambda}_2\|^2 \leq \mu^2 M^2$ if
    $\bm{\lambda}_1=\mathbf{0}$. This completes the proof.

\bibliographystyle{IEEEtran}

\end{document}